\documentclass[conference]{IEEEtran}
\usepackage{amsfonts}
\usepackage{amssymb}
\usepackage{color}
\usepackage{amsmath}
\usepackage{graphicx}
\usepackage{subfigure}
\usepackage{cite}

\usepackage{pifont}
\IEEEoverridecommandlockouts
\usepackage{cite}
\usepackage{amsmath,amssymb,amsfonts}
\usepackage{graphicx}
\usepackage{textcomp}
\usepackage{xcolor}
\usepackage[letterpaper, top=0.75in, bottom=1.09in, left=0.63in, right=0.63in]{geometry}
\usepackage{amsthm}           
\newtheorem{lemma}{Lemma}      
\newtheorem{remark}{Remark}   
\usepackage[ruled]{algorithm2e}
\usepackage{amsmath}

\begin{document}

\title{Aerial Semantic Relay-Enabled SAGIN: Joint UAV Deployment and Resource Allocation}

\author
    {Yanbo Yin, Dingzhu Wen, Changsheng You, XiaoWen Cao, Tat-Ming Lok, and Dusit Niyato
    \thanks{Y. Yin and D. Wen are with the School of Information Science and Technology, ShanghaiTech University, Shanghai 201210, China (e-mail: \{yinyb2023, wendzh\}@shanghaitech.edu.cn). D. Wen is the corresponding author. 
    }
    \thanks{C. You is with the Department of Electronic and Electrical Engineering, Southern University of Science and Technology, Shenzhen 518055, China (e-mail: youcs@sustech.edu.cn).}
    \thanks{X. Cao is with the College of Electronic and Information Engineering, Shenzhen University and Guangdong Provincial Key Laboratory of Future Networks of Intelligence, Shenzhen 518172, China (email: caoxwen@szu.edu.cn) }
    \thanks{T. -M. Lok is with the Department of Information Engineering, The Chinese University of Hong Kong, Hong Kong (e-mail: tmlok@ie.cuhk.edu.hk)}
    \thanks{D. Niyato is with College of Computing and Data Science, Nanyang Technological University, Singapore (e-mail: dniyato@ntu.edu.sg).}

}



\maketitle

\begin{abstract}
Space-Air-Ground Integrated Networks (SAGINs) are pivotal for enabling ubiquitous connectivity in 6G systems, yet they face significant challenges due to severe satellite-to-ground link impairments. Although Unmanned Aerial Vehicles (UAVs) can function as relay nodes to compensate for air-to-ground channel degradation, the satellite-to-UAV link remains a critical bottleneck. Semantic Communication (SemCom) emerges as a promising solution to enhance spectral efficiency by transmitting essential semantic information. This paper proposes a novel multi-cluster UAV-aided SAGIN SemCom architecture that supports both semantic users (SemUsers) and conventional users (ConUsers). While SemCom is employed in the satellite-to-UAV link to improve transmission efficiency, the UAVs implement an intelligent adaptive relay strategy, capable of either directly forwarding semantic data to SemUsers or converting it into bit-level data for ConUsers. Compared to existing similar schemes, this design guarantees the high-efficiency advantages of SemCom while enabling network access for larger coverage area. A joint optimization problem is formulated to maximize the system's sum-rate through coordinated allocation of power, bandwidth, and UAV positions. To address this non-convex problem, we develop an efficient alternating optimization (AO) algorithm, which decomposes the original problem into tractable subproblems. Numerical results demonstrate that the proposed algorithm significantly outperforms baseline schemes in terms of both sum-rate and spectral efficiency across various channel conditions and user distributions, underscoring the importance of joint resource allocation and intelligent UAV deployment.

\end{abstract}


\IEEEpeerreviewmaketitle

\section{Introduction}
Space-Air-Ground Integrated Networks (SAGINs) have emerged as a critical architecture for enabling ubiquitous connectivity in 6G systems, particularly in scenarios with unreliable terrestrial infrastructure \cite{SAGINmag,saginsurvey}. However, direct satellite-to-ground links often experience significant degradation due to substantial path loss over long distances and atmospheric attenuation, resulting in heavy communication overhead \cite{satchannelestimation,satchannelestimation2}. Semantic Communication (SemCom) addresses this issue by transmitting only essential semantic information rather than raw bits, but it imposes high computational demands on users \cite{yinsem}. By deploying unmanned aerial vehicles (UAVs) as relays, path loss can be mitigated \cite{urbanuav1}, and more importantly, UAVs can perform semantic computation tasks on behalf of users lacking computing capabilities, thereby reducing the hardware requirements \cite{yousem}. This paper proposes an innovative SAGIN downlink architecture where UAVs serve as aerial semantic relays, with a joint resource allocation and deployment strategy to optimize system performance. An efficient optimization algorithm is developed to solve the corresponding problem.

\subsection{Related Work}
SAGIN integrates space, air, and ground networks, envisioned to provide access to more users on Earth \cite{zhouSAGIN}, which also aligns with the vision for 6G global connectivity \cite{wen6g}. For instance, by employing aerial vehicles as relays, users in remote areas can be integrated into the network with lower device requirements \cite{saginOnSea} compared to conventional direct satellite links \cite{ship2sat}. Moreover, in urban areas, aerial relays can enhance satellite-to-ground links, which are affected by obstructions from tall buildings and dynamic atmospheric attenuations \cite{urbanuav1}. Among various aerial platforms, UAVs are considered an ideal choice for aerial relays \cite{UAVsagin}. Their long-term operational capability makes them suitable for providing prolonged services. Furthermore, compared to conventional location-fixed relays \cite{zhouris}, their flexible mobility allows adjustment of deployment strategies to optimize service efficiency \cite{urbanuav1}. Numerous studies have been conducted to investigate various problems, such as multi-device association \cite{association}, deployment strategy \cite{ships}, and channel performance analysis \cite{ship2sat}, using various methods including reinforcement learning \cite{reinforcement} and convex optimization \cite{urbanuav1}. However, existing studies predominantly focus on single-cluster scenarios within conventional Shannon-capacity frameworks, lacking consideration of multi-cluster scenarios and more efficient communication approaches. With the number of accessing users increasing, adopting advanced communication paradigms becomes crucial.

In recent years, SemCom has emerged as a promising paradigm for enhancing transmission efficiency \cite{gunduzsemcom}, 
\cite{csgsemcom}. While early research established device-to-device schemes for various data types, including text \cite{qintext}, images \cite{DuSemComImage}, video \cite{videoSemCom}, and multi-modal \cite{yinsem}, recent efforts have focused on extending its coverage and applicability through relay-aided SemCom. In such systems, relays not only expand the service area \cite{relaysemcom4,relaysemcom5} but can also decode semantic information on behalf of resource-constrained users \cite{relaysemcom1,yousem}. To enhance system performance, numerous resource management schemes have been proposed. Corresponding studies have addressed this challenge through methods such as shared probability graphs for latency minimization \cite{youiccresource}, and adaptive quantization paradigms optimized via deep reinforcement learning to balance transmission quality and resource consumption \cite{qinresource}.

Owing to its potential for enhancing efficiency, integrating SemCom into SAGINs represents a promising direction to surpass Shannon capacity. Recent work has begun to explore this synergy, particularly through semantic relays. For instance, \cite{satchannel} investigates a UAV-based semantic relay system, analyzing the trade-off between communication and computation. Meanwhile, advanced AI techniques have been proposed to equip SAGIN nodes with semantic intelligence for functions such as resource allocation \cite{combined1}, routing \cite{combined2}, and task-driven networking \cite{combined3}. However, existing studies remain largely limited to point-to-point schemes within isolated user clusters. In practice, the large coverage of a satellite necessitates serving multiple clusters simultaneously, which requires coordinated resource orchestration across all network layers. Furthermore, most works assume ground users possess sufficient computational capability for semantic decoding, thereby overlooking users with limited resources.

\subsection{Challenges and Contributions}

As noted above, prior work has not adequately characterized multi-cluster, multi-layer SAGINs serving heterogeneous users. To address this gap, this paper formulates a joint resource-allocation problem for a SAGIN composed of a satellite base station, multiple UAV relays, and mixed ground users. The challenge arises from allocating power and bandwidth across competing demands of multiple clusters and network layers. In addition, the deployment of UAVs must be optimized to balance users' channel gains with the aim of maximizing the system sum-rate. These decision variables are tightly coupled and induce nonconvex constraints, rendering the joint optimization computationally challenging. To address these challenges, a resource-allocation and UAV-deployment scheme for SAGIN is proposed. The detailed contributions are summarized as follows:

\begin{enumerate}
    \item \textbf{UAV-Aided Multi-Cluster SemCom Framework:} We propose a novel SemCom framework for SAGIN. It incorporates a UAV-aided multi-cluster architecture and the scenario of hybrid ground users, which include both semantic-capable users and conventional bit-level users, with dedicated transmission schemes designed for each type. This framework maintains the high efficiency of SemCom while extending its service coverage to ground users with limited computational capabilities.
    
    \item \textbf{Joint Resource and UAV Deployment Optimization:} We formulate a joint optimization problem to maximize the system sum-rate through coordinated allocation of transmit power, bandwidth, and UAV positions, subject to constraints that capture resource competition across network layers, clusters, and users. We decompose the original non-convex problem into tractable convex subproblems by introducing auxiliary variables, enabling the development of an efficient algorithm through alternating optimization. The results demonstrate that power and bandwidth allocation exhibit complementary characteristics, where increased allocation of one resource necessitates corresponding increases in the other. Furthermore, we obtain an optimal UAV position as a weighted centroid, with weights dependent on user distances and channel states.
    
    \item \textbf{Performance Evaluation:} We evaluate the proposed scheme through numerical experiments. Results demonstrate that joint optimization of multiple resources and UAV positions achieves superior sum-rate performance compared to schemes lacking full joint optimization. Furthermore, with a fixed number of users, system performance improves as the number of UAV relays increases, since additional UAVs introduce more communication resources into the system. The results validate the superiority of the proposed approach over various baseline schemes, demonstrating its significant potential in enhancing the efficiency of SAGIN systems.
\end{enumerate}

In the remainder of the paper, Section~II introduces the system model and problem formulation, and the optimization problem is formulated in Section~III. The proposed solution algorithm is presented in Section~IV. Numerical results are provided in Section~V, and conclusions are drawn in Section~VI.

\section{System Model}

\subsection{Network Model and Transmission Scheme}

\begin{figure}
    \centering
    \includegraphics[width=0.7\linewidth]{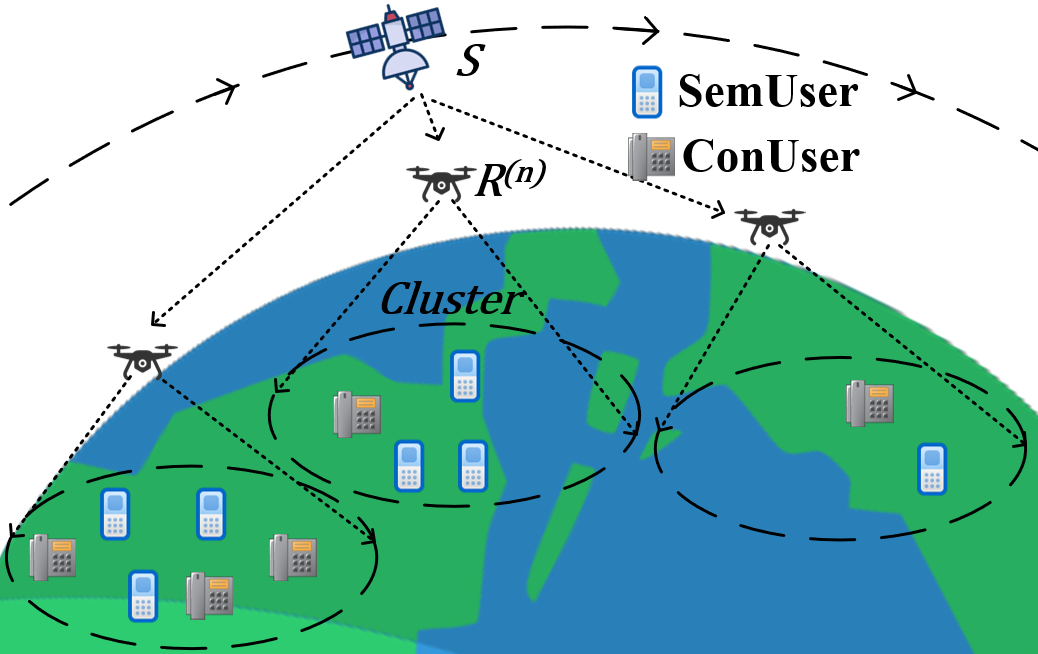}
    \caption{Illustration of the downlink semantic relay system in the considered SAGIN architecture. A satellite serves as the base station in space, while $N$ UAVs function as aerial relays, each supporting a cluster of ground users.}
    \label{fig:system model}
\end{figure}

Consider a downlink SAGIN as illustrated in Fig.~\ref{fig:system model}, comprising a single satellite denoted by $S$ and equipped with a multi-antenna base station (BS), and $N$ single-antenna UAVs serving as aerial relays. Each UAV relay, indexed by $\{R^{(n)}\}_{n=1}^N$, is responsible for forwarding received data from the satellite to a corresponding cluster of single-antenna ground users (GUs). Both the satellite and UAVs are assumed to possess sufficient storage and computational resources to perform both semantic and conventional communication tasks \cite{kbhsagin}. Each ground user cluster, associated with relay $R^{(n)}$, comprises two categories of users:
\begin{itemize}
    \item \textbf{Semantic Users (SemUsers)}: The SemUsers are denoted by $\{U^{(n)}_{\text{sem}, i}\}, \, i \in \mathcal{I}^{(n)} \triangleq \{1, \dots, I^{(n)}\}$, where $I^{(n)}$ is the number of semantic users within the $n$-th cluster. These users are equipped with deep SemCom (DeepSC) decoders and have the computing capabilities to operate in the SemCom mode \cite{yousem}.
    
    \item \textbf{Conventional Users (ConUsers)}: The ConUsers are denoted by $\{U^{(n)}_{\text{con}, j}\}, \, j \in \mathcal{J}^{(n)} \triangleq \{1, \dots, J^{(n)}\}$, where $J^{(n)}$ is the number of conventional users in the $n$-th cluster. These users possess limited processing power and can only operate under traditional bit-level digital communication paradigms \cite{yousem}. 
\end{itemize}

Due to the significant path loss and weak link gains caused by the long-distance satellite-to-ground links, direct communication between the satellite and ground users is considered infeasible. Therefore, UAV-based relaying is employed. The overall system operation, highlighting the semantic relay process, is depicted in Fig.~\ref{fig: workflow}, and is described as below.

\subsubsection{Satellite-to-UAV Hop}
To mitigate the path loss over the satellite-to-UAV links, SemCom is employed in this first-hop transmission. At the satellite, the original data is processed using an AI-driven DeepSC encoder, which extracts high-level semantic features and maps them onto a continuous constellation space\footnote{For practical digital implementation, these continuous symbols can be quantized into a discrete constellation (e.g., high-order QAM).}. These semantic symbols are then transmitted in analog form to each UAV relay $R^{(n)}$. The transmissions from the satellite to the $N$ UAVs are performed over orthogonal frequency bands, thereby eliminating inter-relay interference.

\subsubsection{UAV-to-GU Hop}

Upon reception of the semantic symbols, each UAV relay $R^{(n)}$ operates in a dual-mode transmission strategy tailored to the heterogeneous capabilities of SemUsers and ConUsers in its associated ground cluster:

\begin{itemize}
    \item \textbf{UAV-to-SemUser Transmission} \cite{yousem}: Each relay $R^{(n)}$ forwards the received semantic symbols to its associated SemUsers. Since SemUsers are capable of semantic decoding, the UAV re-encodes the semantic symbols using a conventional digital source and channel encoder into bit streams, which are then transmitted using standard digital modulation schemes. Upon reception, each SemUser performs conventional channel and source decoding to recover the semantic symbols, which are subsequently passed through the locally deployed DeepSC decoder to reconstruct the original message.
    
    \item \textbf{UAV-to-ConUser Transmission} \cite{yousem}: Due to their limited computational capabilities, ConUsers cannot process semantic symbols locally. Hence, the UAV relay decodes the received semantic symbols with onboard DeepSC decoder to reconstruct the original message, which is then re-encoded using conventional source and channel coding schemes and transmitted as bit streams to the ConUsers. Each ConUser applies traditional digital demodulation, channel decoding, and source decoding to recover the original message from the satellite.
\end{itemize}

It is important to note that although SemCom offers significantly higher transmission efficiency, it imposes higher computational demands on both the transmitter and receiver. This trade-off necessitates the strategic deployment of semantic relaying, depending on the computational resources and communication capabilities of the network entities. 
In the considered system, all GUs are concurrently accessed into the network via frequency division multiple access (FDMA), leading to non-overlapping bandwidth assignments for UAVs in the space-to-air layer to mitigate potential interference. Besides, owing to the large between-cluster separation, UAVs do not interfere with users in other clusters, allowing bandwidth reuse across different clusters \cite{relaySemcomResource1}.

\begin{figure*}
    \centering
    \includegraphics[width=0.8\linewidth]{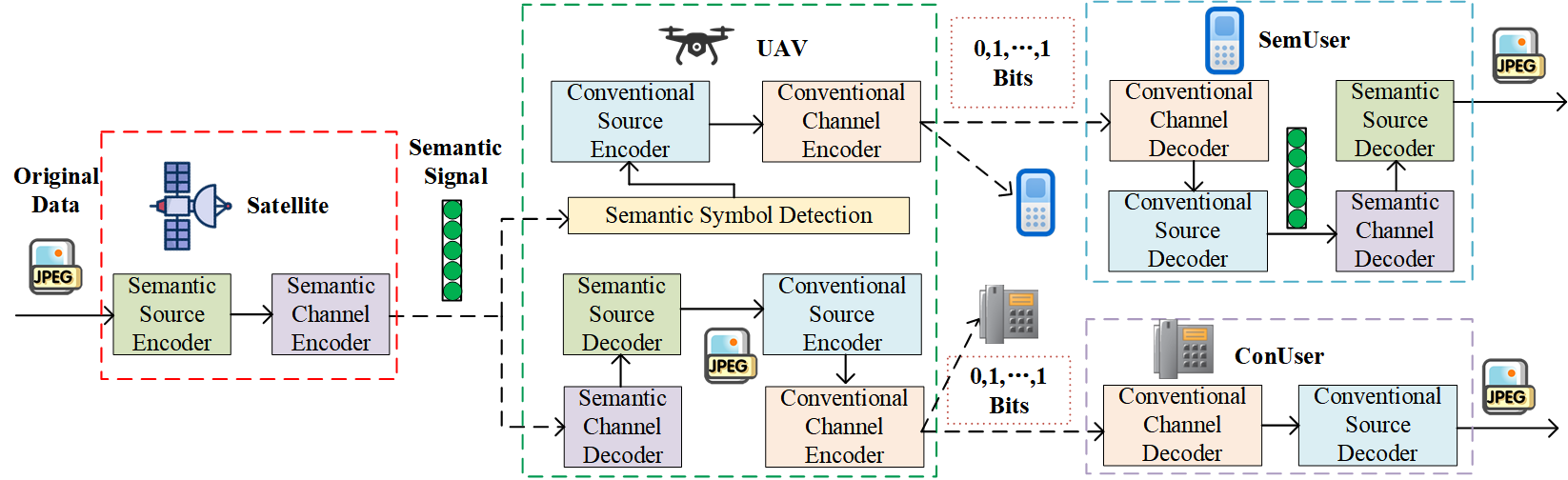}
    \caption{Illustration of dual-mode transmission schemes for semantic-capable and conventional users. The satellite encodes source data into semantic representations. Upon reception, the UAV implements two distinct processing strategies: for semantic users, it forwards semantic symbols directly; for conventional users, it performs semantic decoding and transmits reconstructed bit-level data.}
    \label{fig: workflow}
\end{figure*}

\subsection{Satellite-to-Relay Hop}
The satellite $S$ is equipped with $K$ antennas and a UAV relay $R^{(n)}$ is equipped with a single antenna. The channel vector between $S$ and $R^{(n)}$ is denoted by $\mathbf{h}^{(n)}_{S2R} \in \mathbb{C}^{K \times 1}$ and is influenced by various factors including space-propagation fading, atmospheric absorption, and rain attenuation. Following \cite{satchannel}, the norm of the channel vector is given by $\|\mathbf{h}_{S2R}^{(n)}\|^2 = \left( \frac{\tilde \sigma^{(n)}}{4\pi d^{(n)}_{S2R}} \right)^2$, where $\tilde \sigma^{(n)}$ is the carrier wavelength and $d^{(n)}_{S2R}$ is the Euclidean distance. Following a common approach in multi-input single-output (MISO) channel analysis (see e.g., \cite{scalar_equiv1} and \cite{scalar_equiv2}), the satellite-to-relay channel can be equivalently represented as a scalar channel when appropriate transmit beamforming is applied. This equivalence is justified for system-level performance analysis and resource allocation optimization, as it captures the essential channel characteristics while significantly simplifying the mathematical formulation. It is assumed that the satellite employs maximum ratio transmission (MRT) beamforming towards the UAV relay. The beamforming vector $\mathbf{w}^{(n)} \in \mathbb{C}^{K \times 1}$ is defined with power constraint $\|\mathbf{w}^{(n)}\|^2 = \varsigma^{(n)}$, where $\varsigma^{(n)}$ denotes the satellite beam gain. The effective channel gain is then expressed as:
\begin{equation}
\begin{aligned}
\left| (\mathbf{h}^{(n)}_{S2R})^H \mathbf{w}^{(n)} \right|^2 
&= (\mathbf{h}^{(n)}_{S2R})^H \mathbf{w}^{(n)} (\mathbf{w}^{(n)})^H \mathbf{h}^{(n)}_{S2R} \\
&= \varsigma^{(n)} \left( \frac{\tilde \sigma^{(n)}}{4\pi d^{(n)}_{S2R}} \right)^2,
\end{aligned}
\end{equation}
which corresponds to the scalar channel coefficient in the model. For clarity in subsequent discussions, the channel coefficient is modeled as an equivalent scalar, i.e.,
\begin{equation}
    h^{(n)}_{S2R} = \frac{\sqrt{\varsigma^{(n)}} \tilde \sigma^{(n)}}{4\pi d^{(n)}_{S2R}}.
\end{equation}

The impact of UAV mobility in the satellite-to-UAV link is assumed negligible since the UAV's feasible displacement is several orders of magnitude smaller than $d^{(n)}_{S2R}$ \cite{kbhsagin}. Consequently, the received signal-to-noise ratio (SNR) in dB at $R^{(n)}$ denoted by $r^{(n)}_{S2R}$ is given by
\begin{equation}
r^{(n)}_{S2R}=10 \log_{10} \frac{P^{(n)}_{S2R}\,|h^{(n)}_{S2R}|^2}{B^{(n)}_{S2R}N_0},
\end{equation}
where $B^{(n)}_{S2R}$ is the bandwidth allocated to the link, $P^{(n)}_{S2R}$ is the transmit power assigned by the satellite to this link, and $N_0$ denotes the noise power spectral density.

Following the analysis in \cite{qintext}, a semantic block is defined as the smallest unit conveying one part of complete semantic information, such as one image \cite{imgsimilarity} or one sentence \cite{qintext}. It is assumed that the average number of semantic symbols to represent one semantic block is $Q$, and the average semantic content per data unit is $M$ which is measured in semantic units (suts). Thus, the effective semantic rate in this link, defined by the successfully transmitted semantic information per second (suts/s), can be modeled as 
\begin{equation}
    \Psi^{(n)}_{S2R} = \frac{B^{(n)}_{S2R} M}{Q}\,\varepsilon_Q\big(r^{(n)}_{S2R}\big),
\end{equation}
where $B^{(n)}_{S2R}$ equals the value of semantic symbol rate \cite{yousem} and $\varepsilon_Q(r^{(n)}_{S2R})\in[0,1]$ denotes the semantic similarity between the original and recovered data unit. The function $\varepsilon_Q(\cdot)$ can be approximated by a generalized logistic function parameterized by constants, i.e.,
\begin{equation}
    \varepsilon_Q(r^{(n)}_{S2R}) = a_1 + \frac{a_2}{1 + \exp(-c_1 r^{(n)}_{S2R} - c_2)},
\end{equation}
which can be obtained numerically for a given semantic task and encoder/decoder pair \cite{yousem}. Let $\mu_1$ denote the average number of bits required to represent one data unit in conventional transmission. Since the number of data units conveyed by the transmitted semantic information per second is $\Psi^{(n)}_{S2R}/M$, the effective semantic-to-bit rate for this link can be written as
\begin{equation}
    \Gamma^{(n)}_{S2R} = \mu_1 \frac{\Psi^{(n)}_{S2R}}{M} = \mu_1 \frac{B^{(n)}_{S2R}}{Q}\,\varepsilon_Q\big(r^{(n)}_{S2R}\big), \label{gammas2r}
\end{equation}
which represents the equivalent conventional bit rate (bits/s).

\subsection{Relay-to-Ground Hop}

In practice, the $n$-th UAV is assumed to hover at a fixed minimum height $H_R$ \cite{relaySemcomResource1} to avoid crashing on the terrain, whose position is represented as $\boldsymbol{\ell}^{(n)}_{R}=(x^{(n)}_R, y^{(n)}_R, H_R)$, and the ground users are located on a two-dimensional plane with an elevation of zero. As mentioned above, the transmission strategy from UAV relay to a ground user is tailored to the heterogeneous capabilities of the user, as discussed below.
\begin{itemize}
    \item \textbf{UAV-to-SemUser}: 
For an arbitrary SemUser $U^{(n)}_{sem,i}$ at $\boldsymbol{\ell}^{(n)}_{SU,i}=(x^{(n)}_{SU,i}, y^{(n)}_{SU,i}, 0)$, the Euclidean distance between $R^{(n)}$ and $U_{sem, i}^{(n)}$ is given by
\begin{equation}
    d^{(n)}_{R2SU, i} = \big\lVert\boldsymbol{\ell}^{(n)}_R-\boldsymbol{\ell}^{(n)}_{SU,i}\big\rVert.
\end{equation}
As the channels from UAV to GUs are more likely to be dominated by the LoS link, the channel gain is described by the free-space path loss model \cite{uavchannelnew1}, \cite{uavchannelnew2}, i.e., 
\begin{equation}
\label{channel}
\begin{aligned}
    H^{(n)}_{R2SU, i} =  \beta_0\big(d^{(n)}_{R2SU, i}\big)^{-\alpha} = \beta_0 \big\lVert\boldsymbol{\ell}^{(n)}_R-\boldsymbol{\ell}^{(n)}_{SU,i}\big\rVert^{-\alpha},
\end{aligned}
\end{equation}
where $\beta_0$ denotes the channel power gain at a reference distance $d_0=1\,$m, $\alpha$ is the path-loss exponent. Therefore, the SNR (linear scale) in this link is
\begin{equation}
    \gamma^{(n)}_{R2SU, i} = \frac{H^{(n)}_{R2SU, i} P^{(n)}_{R2SU, i}}{B^{(n)}_{R2SU, i} N_0},
\end{equation}
where $P^{(n)}_{R2SU, i}$ and $B^{(n)}_{R2SU, i}$ are the power and bandwidth allocated by $R^{(n)}$ for transmission to $U^{(n)}_{sem,i}$, respectively. Thus, the spectral efficiency is
\begin{equation}
    \eta^{(n)}_{R2SU, i} = \log_2\big(1+\gamma^{(n)}_{R2SU, i}\big),
\end{equation}
and the corresponding bit-rate is
\begin{equation}
\label{conR2SU}
   C^{(n)}_{R2SU, i} = B^{(n)}_{R2SU, i} \eta^{(n)}_{R2SU, i}.
\end{equation}

Let $\mu_2$ denote the average number of bits required to encode one semantic symbol using conventional source and channel encoders. Hence, the effective semantic-symbol rate in this link, which is the number of semantic symbols successfully transmitted per second, is $C^{(n)}_{R2SU, i}/\mu_2$, and the semantic rate (suts/s) is calculated as
\begin{equation}
\label{semantic rate}
    \Psi^{(n)}_{R2SU, i} = \frac{C^{(n)}_{R2SU, i} M}{\mu_2 Q} = \frac{B^{(n)}_{R2SU, i} \eta^{(n)}_{R2SU, i} M}{\mu_2 Q}.
\end{equation}
Accordingly, the semantic-to-bit rate is
\begin{equation}
\label{gammar2su} 
    \Gamma^{(n)}_{R2SU,i} = \mu_1 \frac{\Psi^{(n)}_{R2SU, i}}{M} = \frac{\mu_1 C^{(n)}_{R2SU, i}}{\mu_2 Q}.
\end{equation}

By comparing (\ref{gammar2su}) and (\ref{conR2SU}), it is evident that, under equivalent bandwidth and power allocations, SemCom can lead to improved communication efficiency relative to conventional transmission.


\item \textbf{UAV-to-ConUser:} For an arbitrary ConUser $U^{(n)}_{con,j}$ at $\boldsymbol{\ell}^{(n)}_{CU,i}=(x^{(n)}_{CU,j}, y^{(n)}_{CU,j}, 0)$, the channel power gain for this $R^{(n)}\to U^{(n)}_{con,j}$ link is calculated as
\begin{equation}
\label{channel_R2CU}
    H^{(n)}_{R2CU,j} = \beta_0 \big\lVert \boldsymbol{\ell}^{(n)}_R - \boldsymbol{\ell}^{(n)}_{CU,j} \big\rVert^{-\alpha},
\end{equation}

and the SNR is calculated as
\begin{equation}
    \gamma^{(n)}_{R2CU,j} = \frac{H^{(n)}_{R2CU,j}P^{(n)}_{R2CU,j}}{B^{(n)}_{R2CU,j}N_0},
\end{equation}
where $B^{(n)}_{R2CU,j}$ and $P^{(n)}_{R2CU,j}$ denote the bandwidth and transmit power allocated by $R^{(n)}$ for transmission to $U^{(n)}_{con,j}$. The spectral efficiency is
\begin{equation}
    \eta^{(n)}_{R2CU,j} = \log_2\big(1+ \gamma^{(n)}_{R2CU,j}\big),
\end{equation}

and the achievable bit rate is
\begin{equation}
\label{gammar2cu}
    \Gamma^{(n)}_{R2CU,j} = B^{(n)}_{R2CU,j} \eta^{(n)}_{R2CU,j},
\end{equation}

\end{itemize}


\section{Problem Formulation}

The objective of the system is to maximize the multiuser sum-rate by jointly  allocating the transmit powers $\boldsymbol{P} \triangleq \{P^{(n)}_S, P^{(n)}_{R2SU, i}, P^{(n)}_{R2CU, j}, \forall n \in \mathcal{N}, \forall i \in \mathcal{I}^{(n)}, \forall j \in \mathcal{J}^{(n)}\}$, the bandwidths $\boldsymbol{B} \triangleq \{B^{(n)}_{S2R}, B^{(n)}_{R2SU, i}, B^{(n)}_{R2CU, j}, \forall n \in \mathcal{N}, \forall i \in \mathcal{I}^{(n)}, \forall j \in \mathcal{J}^{(n)}\}$ and determining the UAV positions $\boldsymbol{L} = \{ 
\boldsymbol{\ell}^{(n)}_R, \forall n \in \mathcal{N} \}$. The objective function is given as
\begin{equation}
    \max_{\boldsymbol{P}, \boldsymbol{B}, \boldsymbol{L}} \quad \sum_{n=1}^N(\sum_{i=1}^{I^{(n)}} \Gamma^{(n)}_{R2SU, i} + \sum_{j=1}^{J^{(n)}} \Gamma^{(n)}_{R2CU, j})  .
\end{equation}
In a practical semantic relay system in a SAGIN, the following constraints need to be considered.
\subsubsection{Transmission Rate Constraint of Relay}
The information forwarded from a relay to all associated ground users cannot exceed the information received from $S$:
    \begin{equation}
        \sum_{i=1}^{I^{(n)}} \Gamma^{(n)}_{R2SU, i} + \sum_{j=1}^{J^{(n)}} \Gamma^{(n)}_{R2CU, j} \le \Gamma ^{(n)}_{S2R}, \forall n \in \mathcal{N},
    \end{equation}
    which, by substituting (\ref{gammas2r}), (\ref{gammar2su}) and (\ref{gammar2cu}), is derived as
    \begin{equation*}
    \mathcal{C}_1:   
    \begin{aligned}
        & \sum_{i=1}^{I^{(n)}} \frac{\mu_1 B^{(n)}_{R2SU, i} \eta^{(n)}_{R2SU, i}}{\mu_2 Q}  + \sum_{j=1}^{J^{(n)}} B^{(n)}_{R2CU, j}  \eta^{(n)}_{R2CU, j} \\
        & \le \mu_1 \frac{B^{(n)}_{S2R} }{Q} \varepsilon_Q(r^{(n)}_{S2R}), \forall n \in \mathcal{N}.
    \end{aligned}
    \end{equation*}

\subsubsection{Bandwidth Constraints} 

As mentioned, the links from $S$ to different relays are performed over different frequency bands. The total bandwidth cannot be beyond the maximum available bandwidth $\tilde{B_S}$, defined as follows:
\begin{equation}
\mathcal{C}_2:        \sum_{n=1}^N B^{(n)}_{S2R} \le \tilde{B_S}.
\end{equation}
On the other hand, for each UAV and its associated user cluster, the total bandwidth consumed by receiving the signal from the satellite and its downlink transmissions to all associated ground users is upperly bounded by the maximum available  bandwidth $\tilde{B_R}$, defined as follows:
    \begin{equation*}
    \mathcal{C}_3:        B^{(n)}_{S2R} + \sum_{i=1}^{I^{(n)}} B^{(n)}_{R2SU, i}+ \sum_{j=1}^{J^{(n)}} B^{(n)}_{R2CU, j} \le \tilde{B_R}, \forall n \in \mathcal{N}.
    \end{equation*}
Note that $B^{(n)}_{S2R}$ appears in both constraints $\mathcal{C}_2$ and $\mathcal{C}_3$, leading to resource competition between different hops, which will be analyzed in detail subsequently.

\subsubsection{Power Constraint} The total transmission power from the satellite to all UAVs must not exceed the maximum power of its transmission module:
    \begin{equation}
    \mathcal{C}_4:    \sum_{n=1}^{N} P_S^{(n)} \le \tilde{P_S}.
    \end{equation}
    
 For each UAV, different communication tasks impose distinct computational loads, and the data processing rate must match the data transmission rate. Following the model in \cite{qinload}, let $G_{sem}$ and $G_{con}$ denote the average number of floating-point operations (FLOPs) required per semantic unit for semantic and conventional transmission, respectively, and let $z$ represent the number of FLOPs per computation cycle. The computation frequency $\nu^{(n)}$ required for UAV $n$ can thus be expressed as:
    \begin{equation}
    \begin{aligned}
        \nu^{(n)} =&\frac{G_{sem}\sum_{i=1}^{I^{(n)}} {\Gamma}^{(n)}_{R2SU,i}}{z}+ \frac{G_{con}\sum_{i=1}^{J^{(n)}} \Gamma^{(n)}_{R2CU, j}}{z} \\
        =& \frac{G_{sem}}{z}\sum_{i=1}^{I^{(n)}} \frac{\mu_1 B^{(n)}_{R2SU, i} \eta^{(n)}_{R2SU, i}}{\mu_2 Q} \\
         & + \frac{G_{con}}{z} \sum_{j=1}^{J^{(n)}} B^{(n)}_{R2CU, j} \eta^{(n)}_{R2CU, j}, \forall n \in \mathcal{N} .
    \end{aligned}
    \end{equation}
    According to \cite{qinload}, the computational power consumption can be modeled as $\zeta_0 \{\nu^{(n)} \}^3$, where $\zeta_0$ is the effective switched capacitance coefficient. 
    
    Because the power consumption for UAV propulsion, such as movement and hovering, is significantly larger than that for communication and computation, it is assumed that the power for communication and computation is supplied by a separate power source \cite{zhangruiuav}. Let the power for downlink transmission to $U^{(n)}_{\text{sem}, i}$ and $U^{(n)}_{\text{con}, j}$ be denoted as $P^{(n)}_{R2SU, i}$ and $P^{(n)}_{R2CU, j}$, respectively. Then, the power constraint can be formulated as 
        \begin{equation}
        \mathcal{C}_5:  
        \begin{aligned}
             &\zeta_0 \{  \nu^{(n)} \}^3 + \sum_{i=1}^{I^{(n)}} P^{(n)}_{R2SU, i} + \sum_{j=1}^{J^{(n)}} P^{(n)}_{R2CU, j} \le \tilde{P}_R, \\
            & \forall n \in \mathcal{N},
        \end{aligned}
        \end{equation}
    where $\tilde{P}_R$ denotes the maximum available power for communication and computation.
    

    Based on the above, the joint resource allocation problem can be formulated as follows.
\begin{equation}
    \begin{aligned}
        \text{(P1)} \max_{\boldsymbol{P}, \boldsymbol{B}, \boldsymbol{L}} \quad & \sum_{n=1}^N(\sum_{i=1}^{I^{(n)}} \Gamma^{(n)}_{R2SU, i} + \sum_{j=1}^{J^{(n)}} \Gamma^{(n)}_{R2CU, j})  \\
        \text{s.t.} \quad & \mathcal{C}_1 \sim \mathcal{C}_5.
        \end{aligned}
        \notag
\end{equation}
Problem (P1) involves non-convex constraints such as $\mathcal{C}_1$ and $\mathcal{C}_5$ with coupled variables, and a non-convex objective function. The variables are high-dimensional, and tightly interdependent, making the problem difficult to tackle.

\section{Joint UAV Deployment and Resource Allocation}

To solve (P1), auxiliary variables are introduced to decompose the problem into three solvable convex subproblems. An efficient algorithm is then developed based on alternating optimization of these subproblems, as detailed below.

\subsection{Introducing Auxiliary Variable Sets}
 To decouple these interdependent variables to solve (P1), the following auxiliary variable sets are introduced, detailed below along with their physical interpretations:
\begin{itemize}
    \item \textbf{User channel gain lower bound set} $\boldsymbol{\hat H}$: Defined as $\{ \hat H^{(n)}_{R2SU, i},  \hat H^{(n)}_{R2CU, j}, \forall j \in \mathcal{J}^{(n)}, \forall i \in \mathcal{I}^{(n)}, \forall n \in \mathcal{N} \}$, where $\hat H^{(n)}_{R2SU, i}$ and $\hat H^{(n)}_{R2CU, j}$ denote the lower bounds of the channel power gains for $R^{(n)} \to U^{(n)}_{sem,i}$ link and $R^{(n)} \to U^{(n)}_{con,j}$ link, respectively. These lower bounds satisfy the following constraint:
    \begin{equation}
    \mathcal{C}_6 :
    \left\{
    \begin{aligned}
      & \hat H^{(n)}_{R2SU, i} \le \beta_0 ||\boldsymbol{\ell}^{(n)}_R-\boldsymbol{\ell}^{(n)}_{SU,i}||^{-\alpha}, \\
      & \hat H^{(n)}_{R2CU, j} \le \beta_0 ||\boldsymbol{\ell}^{(n)}_R-\boldsymbol{\ell}^{(n)}_{CU,j}||^{-\alpha}, \\
      & \forall j \in \mathcal{J}^{(n)}, \forall i \in \mathcal{I}^{(n)}, \forall n \in \mathcal{N}.
    \end{aligned}
    \right.
    \end{equation}
    This constraint is reformulated as the following convex expression with respect to $\boldsymbol{\ell}$:
    \begin{equation}
    \mathcal{C}_6' :
    \left\{
    \begin{aligned}
      & ||\boldsymbol{\ell}^{(n)}_R-\boldsymbol{\ell}^{(n)}_{SU,i}||^2 \le \left[\frac{\beta_0  }{\hat H^{(n)}_{R2SU, i}}\right]^{\frac{2}{\alpha}}, \\
      & ||\boldsymbol{\ell}^{(n)}_R-\boldsymbol{\ell}^{(n)}_{CU,j}||^2 \le \left[\frac{\beta_0  }{\hat H^{(n)}_{R2CU, j}}\right]^{\frac{2}{\alpha}}, \\
      & \forall j \in \mathcal{J}^{(n)}, \forall i \in \mathcal{I}^{(n)}, \forall n \in \mathcal{N}.
    \end{aligned}
    \right.
    \end{equation}

    \item \textbf{User and UAV SNR bound set} $\boldsymbol{\hat \gamma}$: Defined as $\{ \hat{r}^{(n)}_{S2R}, \hat \gamma^{(n)}_{R2SU, i}, \hat \gamma^{(n)}_{R2CU, j}, \forall j \in \mathcal{J}^{(n)}, \forall i \in \mathcal{I}^{(n)}, \forall n \in \mathcal{N} \}$, where $\hat{\gamma}^{(n)}_{S2R}$ represents the lower bound of the SNR for the $S \to R^{(n)}$, while $\hat \gamma^{(n)}_{R2SU, i}$ and $\hat \gamma^{(n)}_{R2CU, j}$ denote the upper bounds of the SNRs for $R^{(n)} \to \hat U^{(n)}_{R2SU, i}$ and $R^{(n)} \to \hat U^{(n)}_{R2CU, j}$ links, respectively. The distinct directions of these bounds are designed to ensure compliance with the relay's transmission rate constraints. The constraints are expressed as:
    \begin{equation}
        \mathcal{C}_7: \hat r^{(n)}_{S2R} \le 10 \log_{10} \frac{ | h^{(n)}_{S2R} |^2 P^{(n)}_{S2R}}{N_0 B^{(n)}_{S2R}}, \forall n \in \mathcal{N},
    \end{equation}
    \begin{equation}
        \mathcal{C}_{8} :
        \left\{
        \begin{aligned}
          & \hat \gamma^{(n)}_{R2SU, i} \ge \frac{ \hat H^{(n)}_{R2SU, i} P^{(n)}_{R2SU, i}}{B^{(n)}_{R2SU, i} N_0}, \\
          & \hat \gamma^{(n)}_{R2CU, j} \ge \frac{ \hat H^{(n)}_{R2CU, j} P^{(n)}_{R2CU, j}}{B^{(n)}_{R2CU, j} N_0}, \\
          & \forall j \in \mathcal{J}^{(n)}, \forall i \in \mathcal{I}^{(n)}, \forall n \in \mathcal{N}.
        \end{aligned}
        \right.
    \end{equation}
    Constraint $\mathcal{C}_7$ is alternatively expressed in convex form with respect to $\boldsymbol{B}$ or $\{\hat r^{(n)}_{S2R} \}$:
    \begin{equation}
    \mathcal{C}_7':
    \begin{aligned}
          & N_0 10^{(r^{(n)}_{S2R} / 10)} (B^{(n)}_{S2R})^2 - | h^{(n)}_{S2R} |^2 P^{(n)}_{S2R} B^{(n)}_{S2R} \le  0,\\
          & \forall n \in \mathcal{N}.
    \end{aligned}
    \end{equation}

    \item \textbf{User spectral efficiency upper bound set} $\boldsymbol{\hat \eta}$: Defined as $\{ \hat \eta^{(n)}_{R2SU, i}, \hat \eta^{(n)}_{R2CU, j}, \forall j \in \mathcal{J}^{(n)}, \forall i \in \mathcal{I}^{(n)}, \forall n \in \mathcal{N} \}$, where $\hat \eta^{(n)}_{R2SU, i}$ and $\hat \eta^{(n)}_{R2CU, j}$ represent the upper bounds of the spectral efficiencies for $R^{(n)} \to \hat U^{(n)}_{R2SU, i}$ and $R^{(n)} \to \hat U^{(n)}_{R2CU, j}$ links, respectively. The constraint is given by:
    \begin{equation}
        \left\{
        \begin{aligned}
          & \hat \eta^{(n)}_{R2SU, i} \ge \log_2( 1 + \hat \gamma^{(n)}_{R2SU, i} ), \\
          & \hat \eta^{(n)}_{R2CU, j} \ge \log_2( 1 + \hat \gamma^{(n)}_{R2CU, j} ), \\
          & \forall j \in \mathcal{J}^{(n)}, \forall i \in \mathcal{I}^{(n)}, \forall n \in \mathcal{N}.
        \end{aligned}
        \right.
    \end{equation}
     This is reformulated as:
    \begin{equation}
    \mathcal{C}_{9}:
        \left\{
        \begin{aligned}
          & 2^{-\hat \eta^{(n)}_{R2SU, i}} -  \frac{1}{1 + \hat \gamma^{(n)}_{R2SU, i}} \le 0,   \\
          & 2^{-\hat \eta^{(n)}_{R2CU, j}} -  \frac{1}{1 + \hat \gamma^{(n)}_{R2CU, j}} \le 0, \\
          & \forall j \in \mathcal{J}^{(n)}, \forall i \in \mathcal{I}^{(n)}, \forall n \in \mathcal{N},
        \end{aligned}
        \right.
    \end{equation}
    or equivalently in convex form with respect to $\boldsymbol{\hat \gamma}$:
    \begin{equation}
        \mathcal{C}_{9}' :
        \left\{
        \begin{aligned}
          &  (\hat \gamma^{(n)}_{R2CU, j})^2 - (2^{\hat \eta^{(n)}_{R2CU, j}} -1) \hat \gamma^{(n)}_{R2CU, j} \le 0,\\
          &  (\hat \gamma^{(n)}_{R2CU, j})^2 - (2^{\hat \eta^{(n)}_{R2CU, j}} -1) \hat \gamma^{(n)}_{R2CU, j} \le 0,\\
          & \forall j \in \mathcal{J}^{(n)}, \forall i \in \mathcal{I}^{(n)}, \forall n \in \mathcal{N}.
        \end{aligned}
        \right.
    \end{equation}

    \item \textbf{UAV operation frequency upper bound set} $\boldsymbol{\hat \nu}$: Defined as $\{ \hat{\nu}^{(n)}, \forall n \in \mathcal{N}  \}$, where $\hat{\nu}^{(n)}$ denotes the upper bound of the operation frequency for UAV $R^{(n)}$, constrained by:
    \begin{equation}
    \mathcal{C}_{10}:
    \begin{aligned}
        \hat \nu^{(n)}  \ge & \frac{G_{sem}}{z} \hat \Gamma^{(n)}_{sum,SU} + \frac{G_{con}}{z} \hat \Gamma^{(n)}_{sum,CU} , \forall n \in \mathcal{N} ,
    \end{aligned}
    \end{equation}
    where $\hat \Gamma^{(n)}_{sum,SU}$ and $\hat \Gamma^{(n)}_{sum,CU}$ are defined by the following equations:
    \begin{equation}
        \hat \Gamma^{(n)}_{sum,SU} = \sum_{i=1}^{I^{(n)}} \frac{\mu_1 B^{(n)}_{R2SU, i} \hat \eta^{(n)}_{R2SU, i}}{\mu_2 Q},
    \end{equation}
    \begin{equation}
        \hat \Gamma^{(n)}_{sum,CU} = \sum_{j=1}^{J^{(n)}} B^{(n)}_{R2CU, j}  \hat \eta^{(n)}_{R2CU, j}.
    \end{equation}
    
\end{itemize}

In addition to constraint non-convexities, the objective function of problem (P1) is non-concave with respect to $\boldsymbol{\ell}$. To address this limitation, the original objective is modified as follows:
\begin{equation}
\begin{aligned}
    \mathcal{F}_{obj} =& \sum_{n=1}^N \sum_{i=1}^{I^{(n)}} \frac{\mu_1 B^{(n)}_{R2SU, i}}{\mu_2 Q} \log_2\left(1 + \frac{  \hat H^{(n)}_{R2SU, i}  P^{(n)}_{R2SU, i}}{B^{(n)}_{R2SU, i} N_0 } \right) \\
    &+ \sum_{n=1}^N \sum_{j=1}^{J^{(n)}} B^{(n)}_{R2CU, j} \log_2\left(1 + \frac{  \hat H^{(n)}_{R2CU, j}  P^{(n)}_{R2CU, j}}{B^{(n)}_{R2CU, j} N_0 } \right) .
\end{aligned}
\end{equation}

In the reformulated problem with objective function $\mathcal{F}_{obj}$, constraints $\mathcal{C}_1$ and $\mathcal{C}_5$ are expressed in terms of the new variable sets, i.e.,
\begin{equation}
    \mathcal{C}_{11}:   
    \begin{aligned}
        & \hat \Gamma^{(n)}_{sum,SU}  + \hat \Gamma^{(n)}_{sum,CU}  \le \mu_1 \frac{B^{(n)}_{S2R} }{Q} \varepsilon_Q(\hat r^{(n)}_{S2R}), \forall n \in \mathcal{N},
    \end{aligned}
    \end{equation}
    \begin{equation}
    \mathcal{C}_{12}: 
    \begin{aligned}
          &\zeta_0 \{ \hat \nu^{(n)}\}^3 + \sum_{i=1}^{I^{(n)}} P^{(n)}_{R2SU, i}+ \sum_{j=1}^{J^{(n)}} P^{(n)}_{R2CU, j} \le \tilde{P_{R}}, \\
        & \forall n \in \mathcal{N} ,
    \end{aligned}
    \end{equation}
    where $\mathcal{C}_{11}$ can be rewritten into the form that is convex on ${\hat r^{(n)}_{S2R}}$
    \begin{equation}
    \mathcal{C}_{11}':   
    \begin{aligned}
       & e^{\varepsilon_Q^{-1}[\frac{Q}{\mu_1 B^{(n)}_{S2R}} (\hat \Gamma^{(n)}_{sum,SU} + \hat \Gamma^{(n)}_{sum,CU})] - \hat r^{(n)}_{S2R}} - 1   \le 0 ,\\
       & \forall n \in \mathcal{N},\\
    \end{aligned}
    \end{equation}
    where $ \varepsilon_Q^{-1}( \cdot)$ is the inverse function of $ \varepsilon_Q( \cdot)$
    

Through the above derivations, problem (P1) is reformulated as the following equivalent problem:
\begin{equation}
    \begin{aligned}
        \text{(P2)} \max_{\substack{
    \boldsymbol{P},\boldsymbol{B},\boldsymbol{L},\\
    \boldsymbol{\hat H},\boldsymbol{\hat \gamma}, \boldsymbol{\eta},\boldsymbol{\hat \nu}
        }} \quad 
        & \mathcal{F}_{obj}  \\
        \text{s.t.} \quad & \mathcal{C}_2 \sim \mathcal{C}_4, \mathcal{C}_6 \sim \mathcal{C}_{12}.
        \end{aligned}
        \notag
\end{equation}


\begin{proof} {(\textbf{Equivalence Between (P1) and (P2))}}
For simplicity and without loss of generality, a simplified model with a single cluster and $J$ conventional users is considered. To simplify notation, we omit the cluster index and define: $B_j = B_{\text{R2CU},j}$, $P_j = P_{\text{R2CU},j}$, $\hat H_j = \hat{H}_{\text{R2CU},j}$, $\hat \gamma_j = \hat{\gamma}_{\text{R2CU},j}$, $\hat \eta_j = \hat{\eta}_{\text{R2CU},j}$, $B_s = B_{\text{S2R}}$, $\hat{r} = \hat{r}_{\text{S2R}}$, and $r = r_{\text{S2R}}$.
First, suppose constraint $\mathcal{C}_6$ is not binding for the $j$-th user. Then, $\hat H_j$ can be increased while reducing $B_j$ without affecting the user's rate, thereby allowing $B_s$ to be increased without violating any constraints and increasing the sum-rate—a contradiction. Hence, $\mathcal{C}_6$ must be tight. Next, consider the constraint chain: 
$\sum_{j=1}^J B_j \log_2\left(1 + \frac{\beta_0 \|\ell_R - \ell_j\|^{-\alpha} P_j}{B_j N_0}\right) = \sum_{j=1}^J B_j \log_2\left(1 + \frac{\hat H_j P_j}{B_j N_0}\right) \le \sum_{j=1}^J B_j \log_2(1 + \hat \gamma_j) \le \sum_{j=1}^J B_j \hat \eta_j \le \mu_1 \frac{B_s}{Q} \varepsilon_Q (\hat{r}) \le \mu_1 \frac{B_s}{Q} \varepsilon_Q (r).$
If any inequality in this chain is not tight, the sum-rate can be further increased without constraint violation, contradicting optimality. Therefore, all constraints must hold with equality, ensuring that the feasible sets of (P1) and (P2) are identical, and thus the problems are equivalent.
\end{proof}

However, (P2) remains non-convex due to variable coupling. To overcome this challenge, (P2) is decomposed into three tractable convex subproblems and leverage an alternating optimization approach to iteratively solve them, as detailed in the following subsections.

\subsection{Subproblem 1: Bandwidth Optimization}
The first sub-problem optimizes the bandwidth allocation $\boldsymbol{B}$ with other variables fixed, formulated as:

\begin{equation}
    \begin{aligned}
        \text{(P3)} \min_{ \boldsymbol{B}} \quad &  - \mathcal{F}_{obj}  \\
        \text{s.t.} \quad & \mathcal{C}_2 \sim \mathcal{C}_4,  \mathcal{C}_6, \mathcal{C}_7', \mathcal{C}_{8} \sim \mathcal{C}_{12}.
        \end{aligned}
        \notag
\end{equation}

Problem (P3) is convex and satisfies Slater's condition, which implies strong duality holds between the primal and dual problems. Consequently, the Karush-Kuhn-Tucker (KKT) conditions provide necessary and sufficient conditions for its optimal solution. The following lemma is derived, where $\lambda_{m}$ denotes the dual variable associated with constraint $\mathcal{C}_m$.

\begin{lemma}
\label{subsol1}
The optimal solution to problem (P3) is 
\begin{equation}
    B_{S 2 R}^{*(n)}=\frac{\Lambda^{*(n)}_{B,S2R}+\lambda_{7^{\prime}}^{*(n)}\left|h_{S 2 R}^{(n)}\right|^{2} P_{S 2 R}^{(n)}+\lambda_{11}^{*(n)} \frac{\mu_{1}}{Q} \varepsilon_{Q}\left(\hat{r}_{S 2 R}^{(n)}\right)}{2 \lambda_{7^{\prime}}^{*(n)} N_{0} 10^{\frac{r_{S 2 R}^{(n)}}{10}}} .
\end{equation}
\begin{equation}
    B_{R2SU,i}^{*(n)}=\sqrt{\frac{\lambda_{8,SU,i}^{*(n)}\hat{H}_{R2SU,i}^{(n)}P_{R2SU,i}^{(n)}}{N_{0}\left[-\frac{\partial \mathcal{F}_{obj}}{\partial B_{R2SU,i}^{*(n)} }+\lambda_{3}^{*(n)}+A_{B,SU,i}^{(n)}\right]}},
\end{equation}
\begin{equation}
    B_{R2CU,j}^{*(n)}=\sqrt{\frac{\lambda_{8,CU,j}^{*(n)}\hat{H}_{R2CU,j}^{(n)}P_{R2CU,j}^{(n)}}{N_0\left[-\frac{\partial \mathcal{F}_{obj}}{\partial B_{R2CU,j}^{*(n)} }+\lambda_3^{*(n)}+A_{B,CU,j}^{(n)}\right]}}, 
\end{equation}
with $\Lambda^{*(n)}_{B,S2R} = -\left(\lambda_{2}^{*}+\lambda_{3}^{*(n)}\right) $, $A_{B,SU,i}^{(n)}=\lambda_{10}^{*(n)}\frac{G_{sem}}{z}\frac{\mu_{1}}{\mu_{2}Q}\hat{\eta}_{R2SU,i}^{(n)}+\lambda_{11}^{*(n)}\frac{\mu_{1}}{\mu_{2}Q}\hat{\eta}_{R2SU,i}^{(n)}$, $A_{B,CU,j}^{(n)}=\lambda_{10}^{*(n)}\frac{G_{con}}{z}\hat{\eta}_{R2CU,j}^{(n)}+\lambda_{11}^{*(n)}\hat{\eta}_{R2CU,j}^{(n)}$, where $\{ \lambda_m^* \}$ are the optimal dual variables that can be obtained via the primal-dual algorithm. 
\end{lemma}

\begin{proof}
    Please refer to Appendix A.
\end{proof}
Since the optimal bandwidth expressions involve themselves implicitly via $\frac{\partial \mathcal{F}_{obj}}{\partial B}$, a closed-form solution may not always be obtainable. In practice, one can use iterative numerical methods such as the bisection method or Newton-Raphson method to solve for $B_{R2SU,i}^{(n)}$ and $B_{R2CU,j}^{(n)}$ given the dual variables and system parameters. This process is embedded within the primal-dual optimization loop to ensure convergence.

\begin{remark}{\textup{(Proposed Bandwidth Allocation on Devices)}}
    The optimal bandwidths exhibit the following trends. First, within each cluster $n$, both $B_{R2SU,i}^{(n)}$ and $B_{R2CU,j}^{(n)}$ grow with the link quality terms $\hat H_{R2SU,i}^{(n)}$ (or $\hat H_{R2CU,j}^{(n)}$) and transmit powers $P_{R2SU,i}^{(n)}$ (or $P_{R2CU,j}^{(n)}$), since better physical channels allow more efficient use of spectrum. Moreover, for identical channel and power, semantic users receive larger $B_{R2SU,i}^{(n)}$ than conventional users. Second, across clusters, the satellite-to-relay bandwidth $B_{S2R}^{(n)}$ increases with the downlink channel gain $|h_{S2R}^{(n)}|^2$ and the semantic quantization utility $\varepsilon_Q(\hat r_{S2R}^{(n)})$, indicating that clusters with better physical conditions or more semantically valuable content are prioritized.  
\end{remark}

Additionally, by analyzing the structure of the constraints and the objective function, the following lemma can be established:

\begin{lemma}
    For the optimal bandwidth allocation $\boldsymbol{B}$, constraint $\mathcal{C}_{11}$ must be tight, i.e.,
\begin{equation}
\begin{aligned}
    &\hat \Gamma^{(n)}_{sum,SU}  + \hat \Gamma^{(n)}_{sum,CU}  = \mu_1 \frac{B_{S2R}^{*(n)} }{Q} \varepsilon_Q(\hat r^{(n)}_{S2R}),\\
\end{aligned}
\end{equation}
which can be rearranged as:
\begin{equation}
   B_{S2R}^{*(n)} = \frac{Q}{\mu_1 \varepsilon_Q(\hat r^{(n)}_{S2R})}\left( \hat \Gamma^{(n)}_{sum,SU}  + \hat \Gamma^{(n)}_{sum,CU} \right).
\end{equation}
\end{lemma}

\begin{proof}
    It is easy to verify that both the left-hand side (LHS) and right-hand side (RHS) of the constraint are monotonically increasing functions of the allocated bandwidth. For the optimal bandwidth allocation, if constraint $\mathcal{C}_{11}$ is not tight, i.e., $\hat \Gamma^{(n)}_{sum,SU}  + \hat \Gamma^{(n)}_{sum,CU}  < \mu_1 \frac{B_{S2R}^{*(n)} }{Q} \varepsilon_Q(\hat r^{(n)}_{S2R}$), then the LHS can always be increased by reallocating more bandwidth to the air-to-ground link while correspondingly reducing the bandwidth allocated to the space-to-air link, without violating any constraints. Since the LHS directly influences the objective function and a higher LHS corresponds to a better objective value, such an adjustment would improve the optimality of the solution. This contradicts the assumption of optimality, implying that $\mathcal{C}_{11}$ must be tight at the optimum.
\end{proof}

\begin{remark}{\textup{(Proposed Bandwidth Allocation in Different Layers)}}
     Constraint $\mathcal{C}_{11}$ reveals that for communication in the $n$-th cluster, the bandwidth allocated to the satellite-to-relay layer and the relay-to-users layer share a common cluster bandwidth constraint. Consequently, these layers compete when the available bandwidth in the cluster is limited. This competition implies that the layer with inferior spectral efficiency is allocated more bandwidth, as the sum-rate across both layers must be balanced at optimality.
\end{remark}

\subsection{Sub-problem 2: Auxiliary variables optimization }
The second sub-problem optimizes the auxiliary variables, including users' spectral efficiency upper bound $\boldsymbol{\hat \eta}$, channel gain lower bound $\boldsymbol{\hat H}$, and UAVs' computation frequency upper bound $\boldsymbol{\hat \nu}$, with other variables fixed. This is formulated as:
\begin{equation}
    \begin{aligned}
        \text{(P4)} \min_{ \boldsymbol{\hat \eta}, \boldsymbol{\hat H},\boldsymbol{\hat \nu}   }   \quad &  - \mathcal{F}_{obj}  \\
        \text{s.t.} \quad & \mathcal{C}_2 \sim \mathcal{C}_4, \mathcal{C}_6 \sim \mathcal{C}_{12}
        \end{aligned}
        \notag
\end{equation} 

\begin{lemma}
\label{subsol2}
    The optimal solution to problem (P4) is derived as:
\begin{equation}
        \hat H^{*(n)}_{R2SU,i} = \frac{1}{A^{*(n)}_{\hat H, SU,i}}
\left[
\frac{\tfrac{\mu_1 B^{(n)}_{R2SU,i}}{\mu_2 Q \ln 2}\,A^{*(n)}_{\hat H, SU,i}}
{\lambda^{*(n)}_{6,SU,i} + \lambda^{*(n)}_{8,SU,i}\,A^{*(n)}_{\hat H, SU,i}}
- 1
\right],
\end{equation}
\begin{equation}
    \hat H^{*(n)}_{R2CU,j} = \frac{1}{A^{*(n)}_{\hat H, CU,j}}
\left[
\frac{\tfrac{B^{(n)}_{R2CU,j}}{\ln 2}\,A^{*(n)}_{\hat H, CU,j}}
{\lambda^{*(n)}_{6,CU,j} + \lambda^{*(n)}_{8,CU,j}\,A^{*(n)}_{\hat H, CU,j}}
- 1
\right],
\end{equation}
\begin{equation}
    \hat \eta^{*(n)}_{R2SU,i}= -\log_2(
\frac{\lambda^{*(n)}_{10}\tfrac{G_{sem}}{z}\tfrac{\mu_1 B^{(n)}_{R2SU,i}}{\mu_2 Q}
+\lambda^{*(n)}_{11}\tfrac{\mu_1 B^{(n)}_{R2SU,i}}{\mu_2 Q}}
{\lambda^{*(n)}_{9,SU,i}\ln2}),
\end{equation}
\begin{equation}
     \hat \eta^{*(n)}_{R2CU,j}= -\log_2(
\frac{\lambda^{*(n)}_{10}\tfrac{G_{con}}{z}B^{(n)}_{R2CU,j}
+\lambda^{*(n)}_{11}B^{(n)}_{R2CU,j}}
{\lambda^{*(n)}_{9,CU,j}\ln2}
),
\end{equation}
\begin{equation}
    \hat \nu^{*(n)} = \Bigl(\frac{\lambda^{*(n)}_{10}}{3\,\lambda^{*(n)}_{12}\,\zeta_0}\Bigr)^{\!1/2},
\end{equation}
with $
A^{*(n)}_{\hat H,SU,i} = \frac{P^{(n)}_{R2SU,i}}{B^{(n)}_{R2SU,i}N_0},
\quad
A^{*(n)}_{\hat H,CU,j} = \frac{P^{(n)}_{R2CU,j}}{B^{(n)}_{R2CU,j}N_0},
$
where $\{\lambda^{*}_{m}\}$ are the optimal dual variables obtained via the primal–dual algorithm.
\end{lemma}

\begin{proof}
    Please refer to Appendix A.
\end{proof}


\subsection{Sub-problem 3: UAV location and system power allocation optimization}
The third sub-problem optimizes the system power allocation $\boldsymbol{P}$, UAV positions $\boldsymbol{\ell}$ and the auxiliary variable SNR bound set $\boldsymbol{\hat \gamma}$, with other given variables, which is reformulated as
\begin{equation}
    \begin{aligned}
        \text{(P5)} \min_{ \boldsymbol{\hat \gamma}, \boldsymbol{ P},\boldsymbol{ L}   }   \quad &  - \mathcal{F}_{obj}  \\
        \text{s.t.} \quad & \mathcal{C}_2 \sim \mathcal{C}_4, \mathcal{C}_6', \mathcal{C}_7, \mathcal{C}_{8}, \mathcal{C}_{9}', \mathcal{C}_{10}, \mathcal{C}_{11}', \mathcal{C}_{12}
        \end{aligned}
        \notag
\end{equation}

Note that problem (P5) is convex and satisfies Slater's condition, implying strong duality between the primal and dual problems. Consequently, the Karush-Kuhn-Tucker (KKT) conditions are applied to solve (P5) optimally. The following lemma is derived, where $\lambda_m^{(n)}$ denotes the dual variable associated with constraint $\mathcal{C}_m$.

\begin{lemma}
\label{subsol3}
    The optimal solution to problem (P5) is derived as
\begin{equation}
    P_{S2R}^{*(n)}
    = \frac{10\,\lambda^{*(n)}_{7}}
           {\lambda^{*}_4\,\ln 10},
\end{equation}
\begin{equation}
    \label{P_R2SU_sol}
P^{*(n)}_{R2SU,i}
= \frac{\frac{\mu_1 B^{(n)}_{R2SU,i}}{\mu_2 Q\ln2}}{\lambda^{*(n)}_{8,SU,i} A^{(n)}_{P,SU,i} + \lambda^{*(n)}_{12}} - \frac{1}{A^{(n)}_{P,SU,i}},
\end{equation}
\begin{equation}
   P^{(n)}_{R2CU,j}
= \frac{\frac{B^{(n)}_{R2SU,i}}{\ln2}}{ \lambda^{*(n)}_{8,CU,j} A^{(n)}_{P,SU,i} + \lambda^{*(n)}_{12}} - \frac{1}{A^{(n)}_{P,CU,j}},
\end{equation}
\begin{equation}
\label{L_R_centroid}
\ell^{*(n)}_R = (\hat x^{*(n)}_R, \hat y^{*(n)}_R, H^{(n)}_R).
\end{equation}
\begin{equation}
    \hat \gamma^{*(n)}_{S2R}
    = E^{(n)} - \ln\!\bigl(\tfrac{\lambda^{*(n)}_{7}}{\lambda^{*(n)}_{11 \prime}}\bigr),
    \quad \forall n,
\end{equation}

\begin{equation}
    \hat \gamma^{*(n)}_{R2SU,i}
    = \frac{\lambda^{*(n)}_{8,SU,i}
            + \lambda^{*(n)}_{9',SU,i}\bigl(2^{\hat\eta^{(n)}_{R2SU,i}} -1\bigr)}
           {2\,\lambda^{*(n)}_{9',SU,i}},
    \,
\end{equation}

\begin{equation}
    \hat \gamma^{*(n)}_{R2CU,j}
    = \frac{\lambda^{*(n)}_{8,CU,j}
            + \lambda^{*(n)}_{9',CU,j}\bigl(2^{\hat\eta^{(n)}_{R2CU,j}} -1\bigr)}
           {2\,\lambda^{*(n)}_{9',CU,j}},
    \
\end{equation}
where $
E^{(n)} = \varepsilon_Q^{-1}\Bigl[\tfrac{Q}{\mu_1 B^{(n)}_{S2R}}(\hat\Gamma^{(n)}_{sum,SU}+\hat\Gamma^{(n)}_{sum,CU})\Bigr],
A^{(n)}_{P,SU,i}=\frac{\hat H^{(n)}_{R2SU,i}}{B^{(n)}_{R2SU,i}N_0},
A^{(n)}_{P,CU,j}=\frac{\hat H^{(n)}_{R2CU,j}}{B^{(n)}_{R2CU,j}N_0},
x^{*(n)}_R = \frac{ \sum_{i=1}^{I^{(n)}} \lambda^{*(n)}_{6',SU,i} x_{SU,i} + \sum_{j=1}^{J^{(n)}} \lambda^{*(n)}_{6',CU,j} x_{CU,j} }{ \sum_{i=1}^{I^{(n)}} \lambda^{*(n)}_{6',SU,i} + \sum_{j=1}^{J^{(n)}} \lambda^{*(n)}_{6',CU,j} }$, $y^{*(n)}_R = \frac{ \sum_{i=1}^{I^{(n)}} \lambda^{*(n)}_{6',SU,i} y_{SU,i} + \sum_{j=1}^{J^{(n)}} \lambda^{*(n)}_{6',CU,j} y_{CU,j} }{ \sum_{i=1}^{I^{(n)}} \lambda^{*(n)}_{6',SU,i} + \sum_{j=1}^{J^{(n)}} \lambda^{*(n)}_{6',CU,j} }
$
and $\{\lambda_m^*\}$ are obtained via the primal–dual algorithm.
\end{lemma}

\begin{proof}
    Please refer to Appendix A.
\end{proof}

\begin{remark}{\textup{(Proposed UAV Location and System Power Allocation)}}
The optimal variables exhibit the following trends. For power allocation within each cluster $n$, $P_{R2SU,i}^{*(n)}$ and $P_{R2CU,j}^{*(n)}$ increase with the allocated bandwidth, indicating greater power allocation to channels with superior conditions or higher semantic importance. For UAV deployment, the optimal UAV position $\mathbf{\ell}^{*(n)}_R$ represents a weighted centroid of user locations, thereby balancing proximity to both SemUsers and ConUsers. 
\end{remark}

\subsection{Overall Alternating Optimization Algorithm}
Based on the solutions of the three convex sub-problems (P3), (P4) and (P5), the alternating optimization approach is adopted to solve (P2), as summarized in Algorithm 1. The algorithm commences with comprehensive initialization to ensure feasibility of (P3) by selecting initial values of $\boldsymbol{B}$, $\boldsymbol{P}$, $\boldsymbol{\ell}$, $\boldsymbol{\hat H}$, $\boldsymbol{\hat \eta}$, $\boldsymbol{\hat \gamma}$, and $\boldsymbol{\hat \nu}$. Subsequently, (P3), (P4) and (P5) are sequentially and iteratively solved by fixing the variables of each other. 

\begin{algorithm}[]
	\caption{Joint Optimization of Bandwidth Allocation, Power Allocation and UAV Placement}\label{Alg:Solution}
	\LinesNumbered
	\KwIn{ $\{|h^{(n)}_{S2R}|, \boldsymbol{\ell}^{(n)}_{SU,i}, \boldsymbol{\ell}^{(n)}_{CU,j}, \forall n \in \mathcal{N}, \forall i \in \mathcal{I}^{(n)}, \forall j \in \mathcal{J}^{(n)}\}$, $\tilde{P}_S$, $\tilde{B}_S$, $\tilde{P}_R$, $\tilde{B}_R$, $\mu_1$, $\mu_2$, $Q$. } 
    \textbf{Initialize} $\boldsymbol{B}$, $\boldsymbol{P}$, $\boldsymbol{L}$, $\boldsymbol{\hat H}$, $\boldsymbol{\hat \eta}$, $\boldsymbol{\hat \gamma}$, $\boldsymbol{\hat \nu}$.\\
    \textbf{Loop}\\   
     \quad Obtain optimal solution to (P3) via Lemma \ref{subsol1}.\\
	\quad Obtain optimal solution to (P4) via Lemma \ref{subsol2}.\\
    \quad Obtain optimal solution to (P5) via Lemma \ref{subsol3}.\\
    \textbf{Until Convergence}.\\
    \KwOut{ $\boldsymbol{B}^{*}$, $\boldsymbol{P}^{*}$, $\boldsymbol{L}^{*}$, $\boldsymbol{\hat \gamma}^{*}$, $\boldsymbol{\hat \eta}^{*}$, $\boldsymbol{\hat H}^{*}$, $\boldsymbol{\hat \nu}^{*}$.}
\end{algorithm}


Algorithm 1 needs to be executed to solve (P2), which is equivalent to (P1), only when the network settings change, including channel state information, user locations, power budgets, bandwidth constraints, and semantic parameters. Particularly, in the case where the network settings vary over different transmission frames, Algorithm 1 should be executed in each frame to solve (P2). The complexity of Algorithm 1 depends on alternatively solving the three convex subproblems (P3), (P4), and (P5), whose computation complexities are all \(\mathcal{O}(K^2)\) with \(K\) being the total number of users. Note that the execution of Algorithm 1 is at the central controller with powerful processors, its computation load can be ignored.

\section{Numerical Result}
\begin{figure*}[h]
  \centering
  \begin{minipage}{0.3\linewidth}
    \centering
    \includegraphics[width=\textwidth]{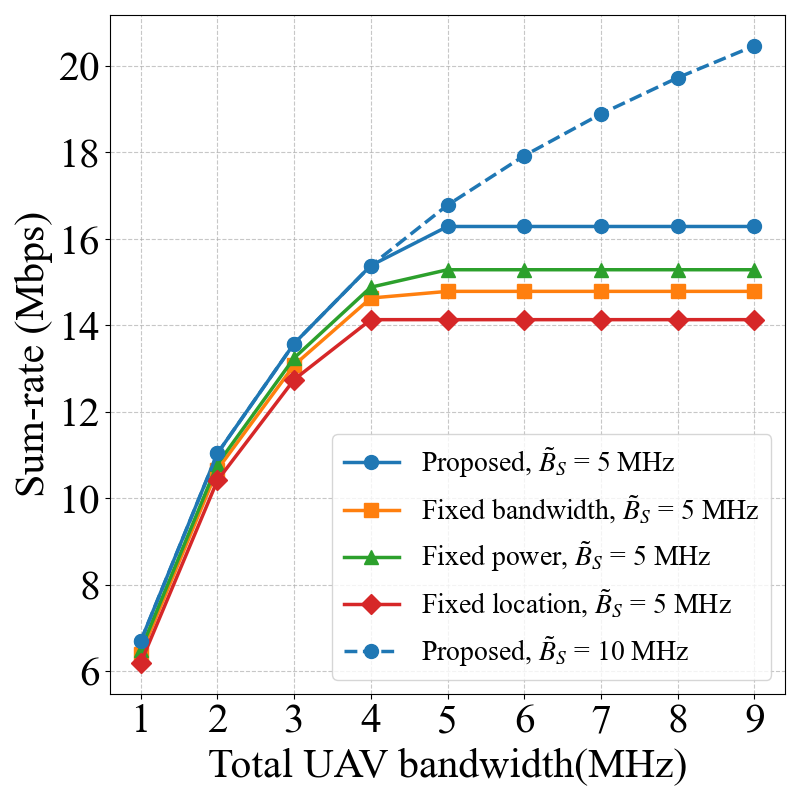} 
    \caption{Sum-rate versus $\tilde{B_R}$ with different algorithms and $\tilde{B_S}$.} 
    \label{fig:s1game1}
  \end{minipage}
  \hfill
  \begin{minipage}{0.3\linewidth}
    \centering
    \includegraphics[width=\textwidth]{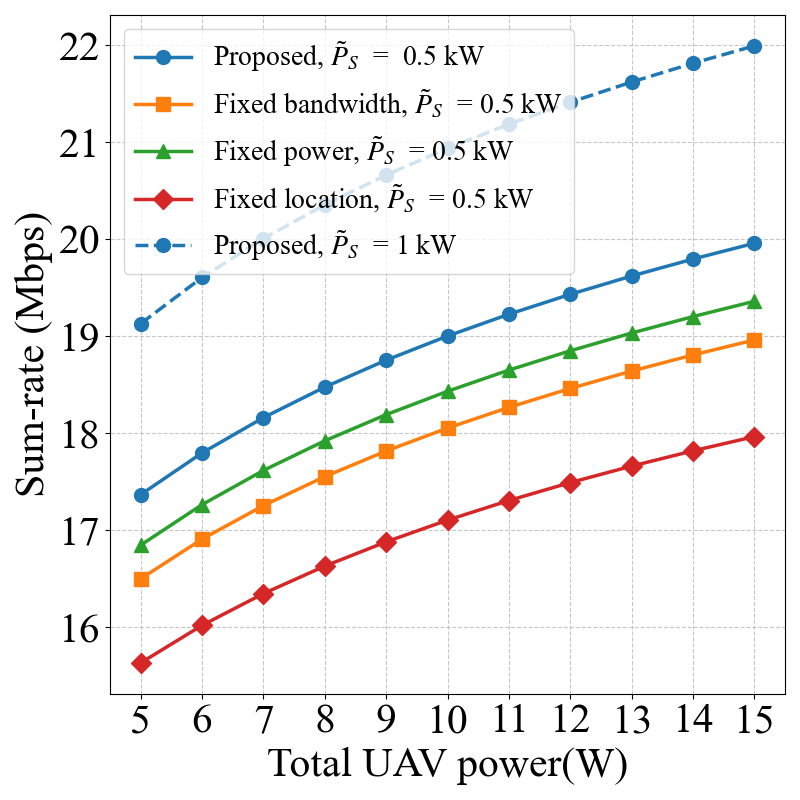} 
    \caption{Sum-rate versus $\tilde{P_R}$ with different algorithms and $\tilde{P_S}$.}
    \label{fig:s1game2}
  \end{minipage}
  \hfill
  \begin{minipage}{0.3\linewidth}
    \centering
    \includegraphics[width=\textwidth]{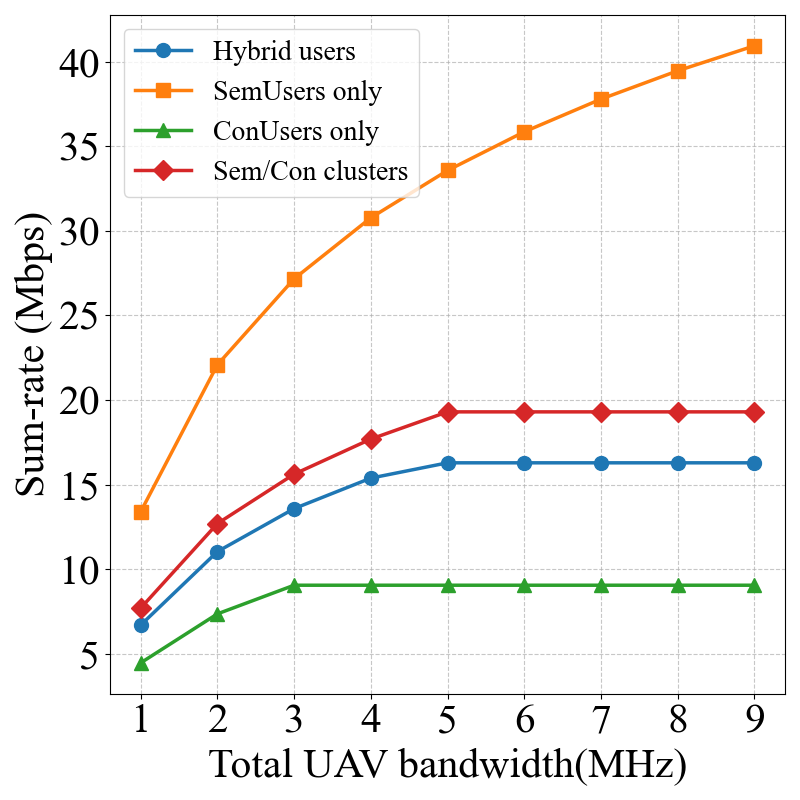} 
    \caption{Sum-rate versus $\tilde{B_R}$ in different user scenarios.}
    \label{fig:s2game1}
  \end{minipage}
\end{figure*}

\begin{figure*}
  \centering
  \begin{minipage}{0.3\linewidth}
    \centering
    \includegraphics[width=\textwidth]{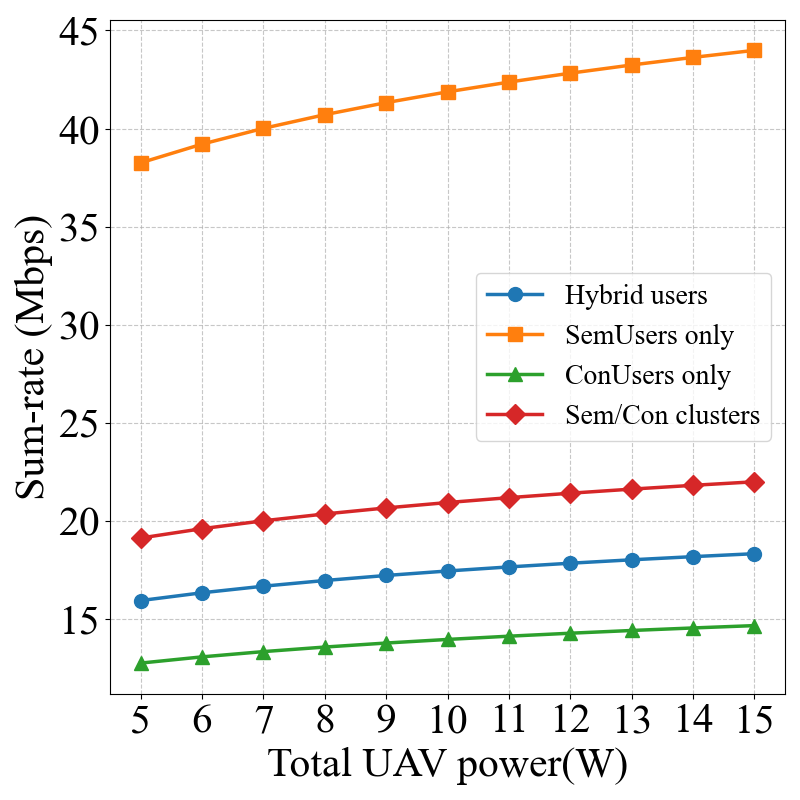} 
    \caption{Sum-rate versus $\tilde{P_R}$ in different user scenarios.} 
    \label{fig:s2game2}
  \end{minipage}
  \hfill
  \begin{minipage}{0.3\linewidth}
    \centering
    \includegraphics[width=\textwidth]{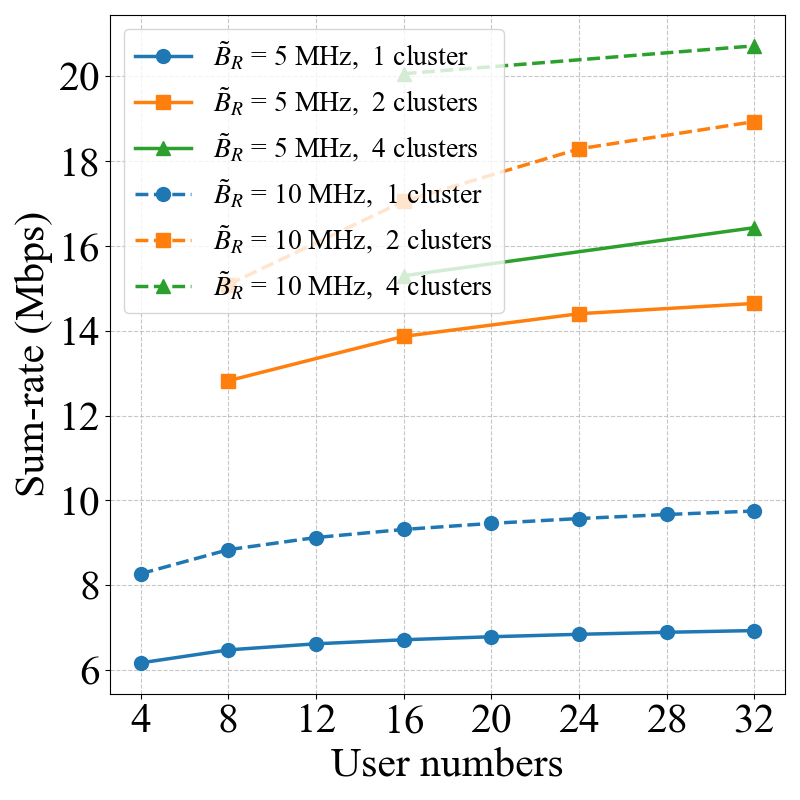} 
    \caption{Sum-rate versus the number of users with different number of clusters and $\tilde{B_R}$.}
    \label{fig:s3game1}
  \end{minipage}
  \hfill
  \begin{minipage}{0.3\linewidth}
    \centering
    \includegraphics[width=\textwidth]{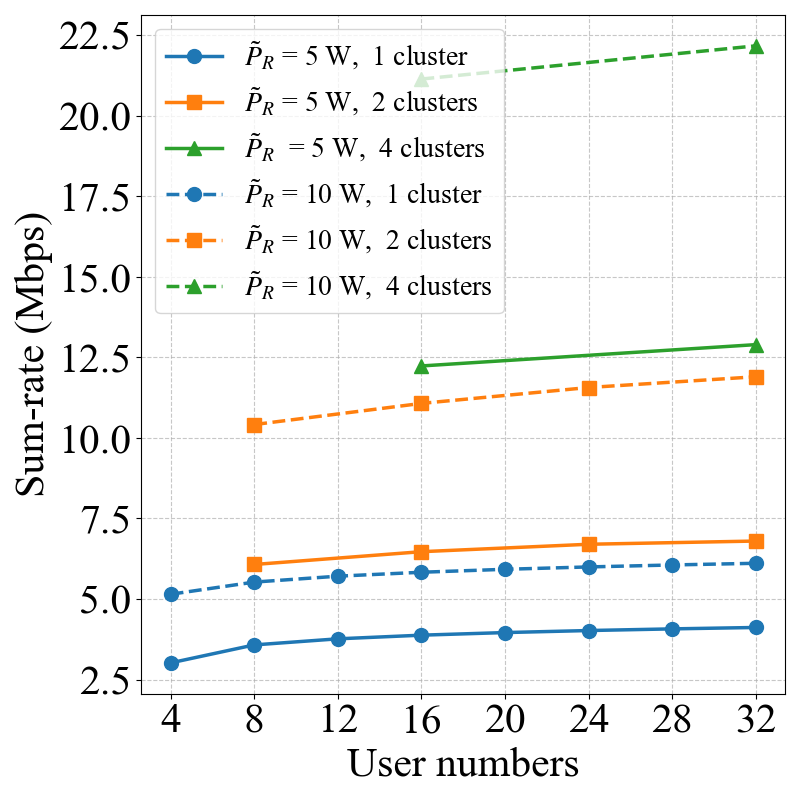} 
    \caption{Sum-rate versus the number of users with different number of clusters and $\tilde{P_R}.$}
    \label{fig:s3game2}
  \end{minipage}
\end{figure*}

In this section, the performance of the proposed system as well as the effectiveness of the proposed algorithm is evaluated. The orbit height of the considered satellite is 60 km and the height of UAVs is 1 km. The noise power spectral density $N_0$ is set to $1 \times 10^{-20}$ W/Hz, corresponding to an extremely low noise floor \cite{yousem}. For the coding of one semantic block, the parameters $\mu_1$, $Q$ and $\mu_2$ are 48, 4, and 4 respectively \cite{textsemcomallocation},\cite{yousem}. For the corresponding parameters of the semantic similarity function, they can be numerically obtained as $a_1 = 0.3980$, $a_2 = 0.5385$, $c_1 = 0.2815$ and $c_2 = -1.3135$ \cite{a1a2c1c2}. As semantic relay involves less computation load than conventional relay, the average number of floating point operations required for semantic and conventional transmission $G_{sem}$ and $G_{con}$ is 2 and 4 respectively, and the number of floating-point operations per computation cycle is 2, the computation power coefficient $\zeta_0$ is $10^{-3}$ \cite{qinload}. For the channel gain in the air-to-ground channel, $\beta_0$ is -60 dB, $\alpha$ is 2 \cite{uavchannel}. Following the similar parameter settings in \cite{satchannel}, the order of magnitude of the channel coefficient in a space-to-air link is -80 dB. In the subsequent simulations, the system parameters are set as follows unless otherwise specified: The satellite bandwidth and power are fixed at $\tilde{B_S} = 10$ MHz and $\tilde{P_S} = 1$ kW, respectively, while the UAV bandwidth and power are set to $\tilde{B_R} = 10$ MHz and $\tilde{P_R} = 10$ W, respectively. In each cluster, users are uniformly distributed within a 1-kilometer square area.

\subsection{Effectiveness of the Proposed Algorithm}

\begin{figure*}[t]
  \centering
  \begin{minipage}{0.3\linewidth}
    \centering
    \includegraphics[width=\textwidth]{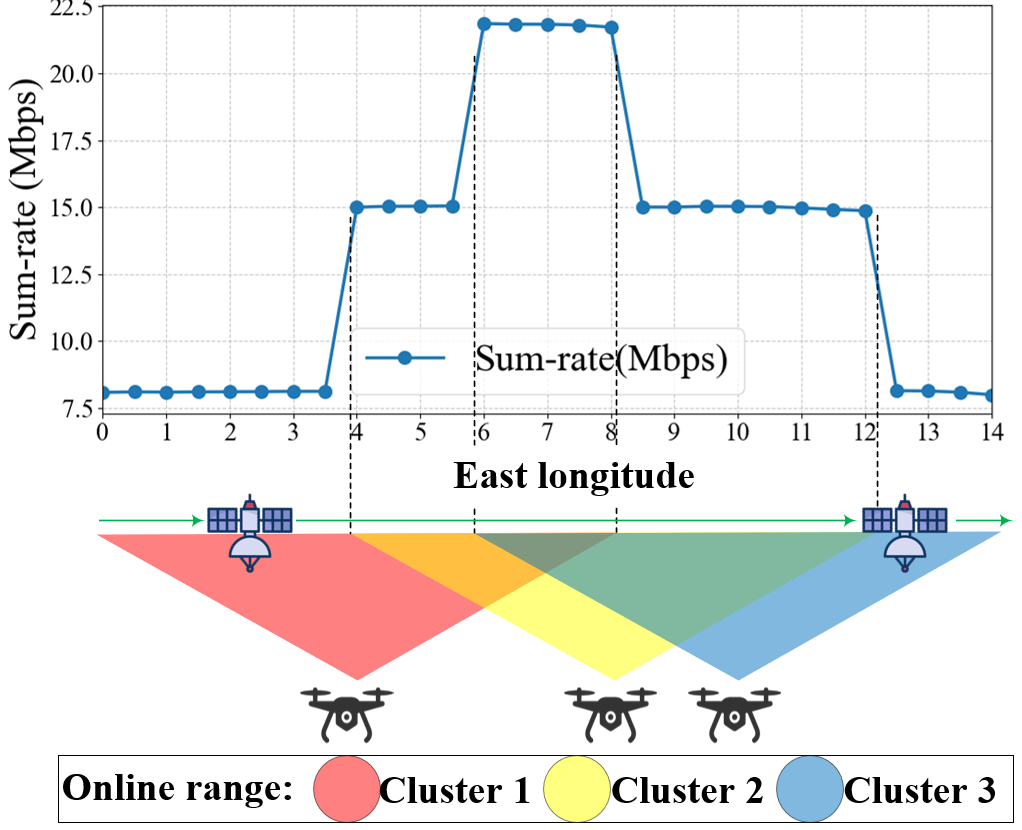} 
    \caption{Long-term system performance during satellite movement} 
    \label{fig:s4}
  \end{minipage}
  \hfill
  \begin{minipage}{0.3\linewidth}
    \centering
    \includegraphics[width=\textwidth]{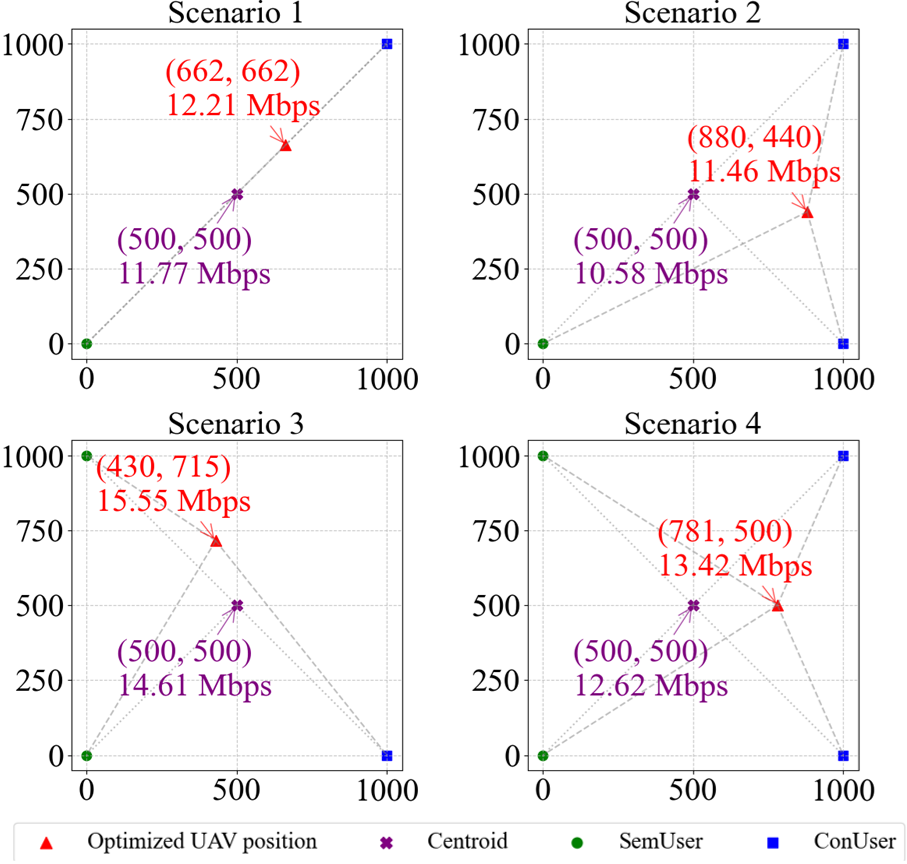} 
    \caption{Optimized UAV position versus different user distribution in regular scenarios.}
    \label{fig:s5}
  \end{minipage}
  \hfill
  \begin{minipage}{0.3\linewidth}
    \centering
    \includegraphics[width=\textwidth]{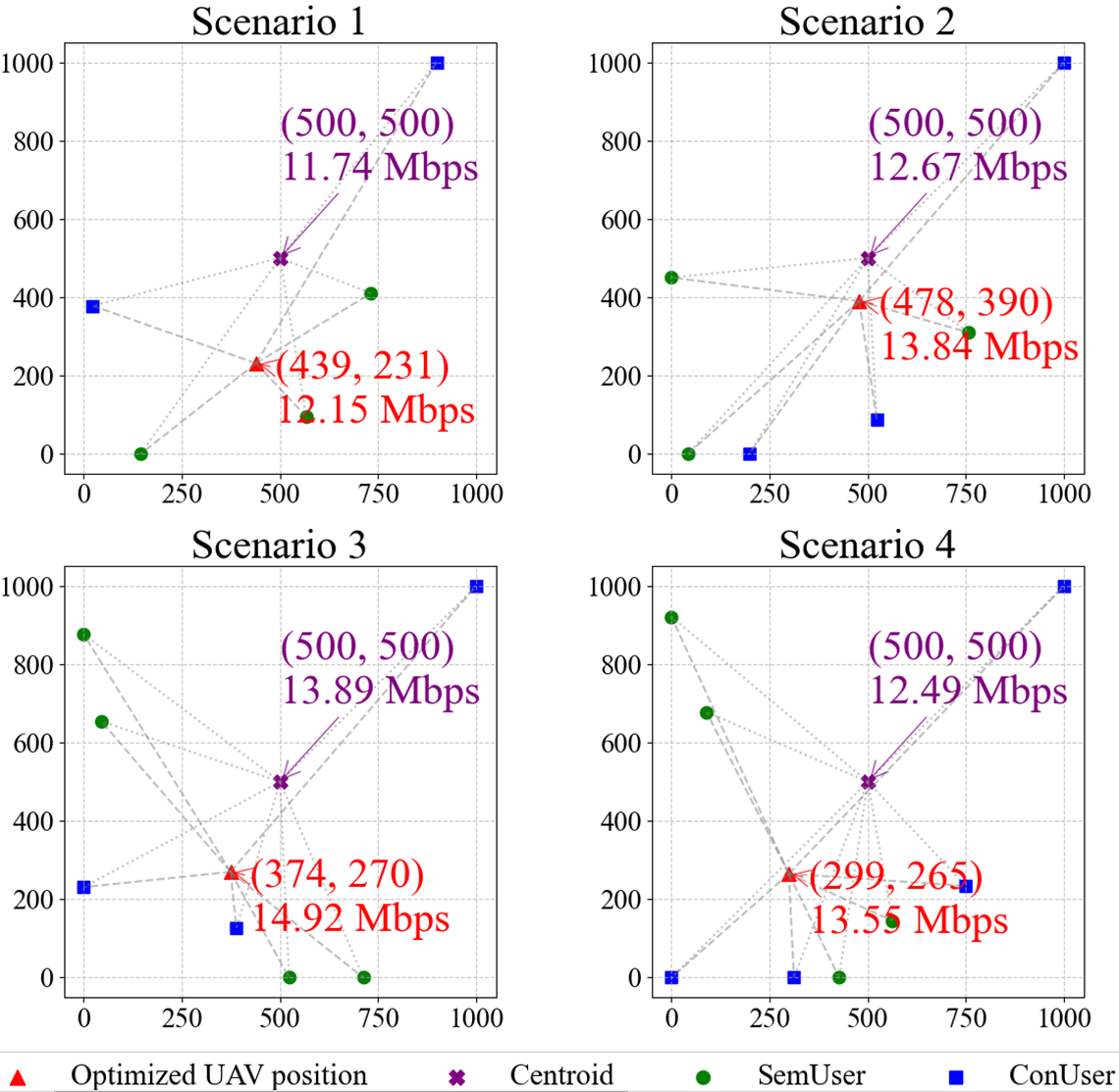} 
    \caption{Optimized UAV position versus different user distribution in random scenarios}
    \label{fig:s5new}
  \end{minipage}
\end{figure*}

In this experiment, a system comprising 5 clusters is considered, each consists of 4 SemUsers and 4 ConUsers. The evaluated algorithms are listed the following:

\begin{enumerate}
    \item \textit{Proposed}: Jointly optimizes power allocation, bandwidth allocation, and UAV positions with proposed algorithm.
    \item \textit{Fixed bandwidth}: Optimizes power allocation and UAV positions under equal bandwidth allocation.
    \item \textit{Fixed power}: Optimizes bandwidth allocation and UAV positions under equal power allocation.
    \item \textit{Fixed location}: Optimizes power and bandwidth allocation with fixed UAV positions.
\end{enumerate}

The results are shown in Fig.~\ref{fig:s1game1} and Fig.~\ref{fig:s1game2}. In Fig.~\ref{fig:s1game1}, the sum-rate is seen to grow quickly with the increase of the total UAV bandwidth $\tilde{B}_R$ at small values and to plateau thereafter. This behaviour is attributed to the competition between the UAV's satellite-facing and user-facing links: when the UAV bandwidth bound is relaxed, contention at the UAV is reduced and end-to-end throughput is improved. However, once $\tilde{B}_R$ exceeds a certain point, further increases no longer translate into system-level gain because the satellite bandwidth $\tilde{B}_S$ becomes the limiting resource. In other words, the UAV-layer bandwidth can only improve performance up to the point at which the satellite-layer bandwidth forms a bottleneck; additional benefit would be realized only if $\tilde{B}_S$ is increased. Across the tested settings, the proposed algorithm consistently attains higher sum-rate than the baseline methods owing to its joint allocation of bandwidth and other resources. 

In Fig.~\ref{fig:s1game2}, the sum-rate increases with the total UAV transmit power $\tilde{P}_R$, while exhibiting diminishing returns at higher power levels. Again, when the satellite power $\tilde{P}_S$ is increased, the performance on the same $\tilde{P}_R$ is improved. These findings collectively indicate that system performance is jointly constrained by resources at both layers. Hence, meaningful throughput enhancement requires coordinated increases of UAV- and satellite-level budgets or an optimization that explicitly addresses the inter-layer competition.

\subsection{Effect of User Scenarios}

This experiment investigates the impact of different user distributions on system performance. The system comprises 4 clusters. The proposed algorithm is applied for resource allocation under the following user arrangements:

\begin{enumerate}
    \item \textit{Hybrid users}: 2 clusters contain 2 SemUsers and 6 ConUsers each, the other 2 clusters contain 6 SemUsers and 2 ConUsers each.
    \item \textit{SemUsers only}: Each cluster consists of 8 SemUsers.
    \item \textit{ConUsers only}: Each cluster consists of 8 ConUsers.
    \item \textit{Sem/Con clusters}: 2 clusters contain 8 ConUsers each, the other 2 clusters contain 8 SemUsers each.
\end{enumerate}

Fig.~\ref{fig:s2game1} and Fig.~\ref{fig:s2game2} illustrate the achieved sum-rate under varying UAV bandwidth and power, respectively. Several important conclusions can be drawn:

First, the \textit{SemUsers only} scenario significantly outperforms all other configurations, achieving the highest sum-rate in both bandwidth and power variation tests. The \textit{ConUsers only} scenario exhibits the lowest sum-rate, and the performance gap between semantic and conventional users emphasizes the critical advantage of semantic technologies in future wireless networks. Notably, the \textit{Heterogeneous users} and \textit{Sem/Con clusters} scenarios demonstrate intermediate performance. However, the \textit{Sem/Con clusters} arrangement slightly outperforms the heterogeneous case under equal user composition. This suggests that grouping users of the same type within clusters improves resource allocation efficiency, likely due to reduced interference and more homogeneous channel conditions within each cluster. The results strongly advocate for user-type-aware clustering strategies in multi-user SemCom systems.

\subsection{Effect of Number of Users and Clusters}

This experiment investigates the impact of user population and cluster configuration on system performance. The proposed algorithm is employed to evaluate the following configurations:

\begin{enumerate}
    \item \textit{Varying User Numbers}: System performance under different user populations (4 to 32 users)
    \item \textit{Varying Cluster Numbers}: System performance under different cluster configurations (1, 2, and 4 clusters)
\end{enumerate}

The results demonstrate that increasing system resources significantly enhances sum-rate performance, particularly through additional UAV relays, while user population growth alone yields limited gains. This underscores the critical role of UAV deployment in overcoming resource constraints. As shown in Figs. \ref{fig:s3game1} and \ref{fig:s3game2}, elevating UAV bandwidth or transmission power consistently improves sum-rate across all scenarios. Furthermore, systems with more UAV clusters achieve superior performance by distributing users among parallel relays, reducing per-user resource competition and improving spatial reuse. These findings validate the essential advantage of UAV-aided relay architectures over direct satellite-to-ground links.

\subsection{Long-term System Performance During Satellite Movement}

This experiment evaluates the system performance during extended satellite operation as the satellite moves from $0 ^{\circ}$ E to $14^{\circ}$ E longitude. Fig.~\ref{fig:s4} presents the achieved sum-rate across the satellite's trajectory.

From Fig.~\ref{fig:s4}, the sum-rate demonstrates significant variation throughout the satellite's flight path. Notably, the system achieves peak performance when the satellite is positioned in the overlapping coverage areas of all three clusters, indicating that spatial diversity and coordinated resource allocation across multiple clusters substantially enhance overall system capacity. These results demonstrate the algorithm's effectiveness in leveraging spatial diversity through intelligent resource allocation across multiple clusters. The progressive performance improvement with increasing cluster accessibility underscores the algorithm's capability to dynamically optimize resource distribution based on real-time cluster availability.

\subsection{UAV Deployment Optimization Analysis}

Fig.~\ref{fig:s5} and Fig.~\ref{fig:s5new} present the optimized UAV positions under different user distributions. The results demonstrate that optimized positioning significantly enhances system performance compared to geometric centroid placement. The UAVs are strategically placed closer to conventional users to compensate for their poorer channel conditions, thereby improving overall rate performance. Furthermore, the proposed algorithm adapts effectively to varying user densities and spatial configurations, confirming its robustness and practical applicability in real-world UAV-assisted communication systems.

\section{Conclusion}
This paper addresses the critical challenge of resource allocation in multi-cluster SAGINs serving heterogeneous users with diverse computational capabilities. We have proposed a novel architecture where UAV relays employ dual-mode transmission to serve both semantic and conventional users. To maximize the sum-rate, we have formulated a joint optimization problem involving power allocation, bandwidth distribution, and UAV deployment. We have solved the resulting non-convex problem through an efficient AO-based algorithm that decouples the problem into tractable convex subproblems via auxiliary variables. Numerical results have demonstrated the superiority of the proposed scheme over various baselines, highlighting how integrated semantic communication and intelligent resource orchestration can enhance next-generation wireless networks.

\bibliography{{reference}}

\section*{Appendix A}
The proof of Lemma~\ref{subsol1} is as follows. As (P3) is convex and satisfies the Slater's condition, the strong duality holds between the primal problem and its dual problem. The KKT conditions are thus applied to solve (P3), where the Lagrangian function is
\begin{equation}
        \begin{aligned}
        \mathcal{L}_1 &= -\mathcal{F}_{obj} + \lambda_2 (\sum_{n=1}^N B^{(n)}_{S2R} - \tilde{B_S})\\
        & + \sum_{n=1}^N\lambda_3^{(n)} (B^{(n)}_{S2R} + \sum_{i=1}^{I^{(n)}} B^{(n)}_{R2SU, i}+ \sum_{j=1}^{J^{(n)}} B^{(n)}_{R2CU, j} - \tilde{B_R})\\
        & + \sum_{n=1}^N \lambda^{(n)}_{7'}(N_0 10^{(r^{(n)}_{S2R} / 10)} (B^{(n)}_{S2R})^2 - | h^{(n)}_{S2R} |^2 P^{(n)}_{S2R} B^{(n)}_{S2R})\\
        & + \sum_{n=1}^N \sum_{i=1}^{I^{(n)}}( \lambda^{(n)}_{8,SU,i} (  \frac{ \hat H^{(n)}_{R2SU, i} P^{(n)}_{R2SU, i}}{B^{(n)}_{R2SU, i} N_0} - \hat \gamma^{(n)}_{R2SU, i} ) )\\
        & + \sum_{n=1}^N \sum_{j=1}^{J^{(n)}}( \lambda^{(n)}_{8,CU,j} (  \frac{ \hat H^{(n)}_{R2CU, j} P^{(n)}_{R2CU, j}}{B^{(n)}_{R2CU, j} N_0} - \hat \gamma^{(n)}_{R2CU, j} ) )\\
        & + \sum_{n=1}^N \lambda^{(n)}_{10} (\frac{G_{sem}}{z} \hat \Gamma^{(n)}_{sum,SU} + \frac{G_{con}}{z} \hat \Gamma^{(n)}_{sum,CU} - \hat \nu^{(n)}) \\
        & + \sum_{n=1}^N \lambda^{(n)}_{11} (\hat \Gamma^{(n)}_{sum,SU}  + \hat \Gamma^{(n)}_{sum,CU} -\mu_1 \frac{B^{(n)}_{S2R} }{Q} \varepsilon_Q(\hat r^{(n)}_{S2R}))
    \end{aligned}
\end{equation}
where $\{\lambda^{(n)}_{m} \}$ are the dual variables associated with the constraints $\mathcal{C}_{m}$ respectively.

Based on the KKT conditions with given $\{\lambda^{(n)}_{m} \}$, it follows that
\begin{equation}
\label{LBS2R}
    \begin{aligned}
        \frac{\partial \mathcal{L}_{1}}{\partial B_{S 2 R}^{(n)}} =& \lambda_{7^{\prime}}^{(n)}\left(2 N_{0} 10^{\frac{\gamma_{S 2 R}^{(n)}}{10}} B_{S 2 R}^{(n)}-\left|h_{S 2 R}^{(n)}\right|^{2} P_{S 2 R}^{(n)}\right)\\
        & + \lambda_{2}+\lambda_{3}^{(n)}
        -\lambda_{11}^{(n)} \frac{\mu_{1}}{Q} \varepsilon_{Q}\left(\hat{\gamma}_{S 2 R}^{(n)}\right)=0 ,\forall n \in \mathcal{N},
    \end{aligned}
\end{equation}
\begin{equation}
\label{LBR2SU}
    \begin{aligned}
        \frac{\partial \mathcal{L}_1}{\partial B_{R2SU,i}^{(n)} } =& -\frac{\partial \mathcal{F}_{obj}}{\partial B_{R2SU,i}^{(n)} }  -\lambda_{8,SU,i}^{(n)}\frac{\hat{H}_{R2SU,i}^{(n)}P_{R2SU,i}^{(n)}}{N_{0}\left(B_{R2SU,i}^{(n)}\right)^{2}} \\
        &+A_{B,SU,i}^{(n)} + \lambda_{3}^{(n)} =0, \forall n \in \mathcal{N},
    \end{aligned}
\end{equation}
\begin{equation}
\label{LBR2CU}
    \begin{aligned}
        \frac{\partial \mathcal{L}_1}{\partial B_{R2CU,j}^{(n)} } =& -\frac{\partial \mathcal{F}_{obj}}{\partial B_{R2CU,j}^{(n)} }  -\lambda_{8,CU,j}^{(n)}\frac{\hat{H}_{R2CU,j}^{(n)}P_{R2CU,j}^{(n)}}{N_{0}\left(B_{R2CU,j}^{(n)}\right)^{2}}\\
        &+A_{B,CU,j}^{(n)} + \lambda_{3}^{(n)} =0, \forall n \in \mathcal{N},
    \end{aligned}
\end{equation}
where $A_{B,SU,i}^{(n)}=\lambda_{10}^{(n)}\frac{G_{sem}}{z}\frac{\mu_{1}}{\mu_{2}Q}\hat{\eta}_{R2SU,i}^{(n)}+\lambda_{11}^{(n)}\frac{\mu_{1}}{\mu_{2}Q}\hat{\eta}_{R2SU,i}^{(n)}$, $A_{B,CU,j}^{(n)}=\lambda_{10}^{(n)}\frac{G_{con}}{z}\hat{\eta}_{R2CU,j}^{(n)}+\lambda_{11}^{(n)}\hat{\eta}_{R2CU,j}^{(n)}$, and (\ref{LBS2R}), (\ref{LBR2SU}) and (\ref{LBR2CU}) are the first-order derivatives of $\mathcal{L}_1$ w.r.t. $B_{S 2 R}^{(n)}$, $B_{R2SU,i}^{(n)}$ and $B_{R2CU,j}^{(n)}$ respectively. By combining with some manipulations, we have
\begin{equation}
    B_{S 2 R}^{(n)}=\frac{\Lambda^{(n)}_{B,S2R}+\lambda_{7^{\prime}}^{(n)}\left|h_{S 2 R}^{(n)}\right|^{2} P_{S 2 R}^{(n)}+\lambda_{11}^{(n)} \frac{\mu_{1}}{Q} \varepsilon_{Q}\left(\hat{\gamma}_{S 2 R}^{(n)}\right)}{2 \lambda_{7^{\prime}}^{(n)} N_{0} 10^{\frac{\gamma_{S 2 R}^{(n)}}{10}}} ,
\end{equation}
\begin{equation}
    B_{R2SU,i}^{(n)}=\sqrt{\frac{\lambda_{8,SU,i}^{(n)}\hat{H}_{R2SU,i}^{(n)}P_{R2SU,i}^{(n)}}{N_{0}\left[ - \frac{\partial \mathcal{F}_{obj}}{\partial B_{R2SU,i}^{(n)} }+\lambda_{3}^{(n)} + A_{B,SU,i}^{(n)}\right]}},
\end{equation}
\begin{equation}
    B_{R2CU,j}^{(n)}=\sqrt{\frac{\lambda_{8,CU,j}^{(n)}\hat{H}_{R2CU,j}^{(n)}P_{R2CU,j}^{(n)}}{N_0\left[-\frac{\partial \mathcal{F}_{obj}}{\partial B_{R2CU,j}^{(n)} }+\lambda_3^{(n)}+A_{B,CU,j}^{(n)}\right]}},
\end{equation}
where $\Lambda^{(n)}_{B,S2R} = -\left(\lambda_{2}+\lambda_{3}^{(n)}\right) $. By leveraging the primal-dual method, it can be updated both the primal variables and the dual variables to approximately satisfy the optimality conditions. At each iteration $k$, dual variables are updated as
\begin{equation}
    \lambda^{(n)}_{m,(k+1)} = \lambda^{(n)}_{m,(k+1)} + \delta^{(n)}_{m,(k+1)} \frac{\partial \mathcal{L}_1}{\partial \lambda^{(n)}_{m,(k+1)}}
\end{equation}
where $\delta^{(n)}_{m,(k+1)}$ is the step size for $\lambda^{(n)}_{m,(k+1)}$. It gradually increases the value of the dual variable until convergence, and then obtains the primary optimal solution through optimal dual variables. This Lemma 1 is thus proved.

Regarding Lemma~\ref{subsol2} and Lemma~\ref{subsol3}, since the corresponding subproblems share similar characteristics with (P3), their proofs follow a comparable calculation process and update method. Hence, the proofs of Lemma~\ref{subsol2} and Lemma~\ref{subsol3} can be established.

\bibliographystyle{IEEEtran}

\end{document}